\documentclass[11pt]{article}
\usepackage[utf8]{inputenc}
\usepackage[margin=1in]{geometry}
\usepackage{amsmath, amssymb, amsthm, dsfont}
\usepackage{hyperref}
\usepackage{natbib}
\setcitestyle{authoryear,open={(},close={)}} 

\usepackage{booktabs} 
\usepackage[ruled]{algorithm2e} 

\SetAlFnt{\small}
\SetAlCapFnt{\small}
\SetAlCapNameFnt{\small}
\SetAlCapHSkip{0pt}
\IncMargin{-\parindent}

 \usepackage[english]{babel}
\usepackage{bm}
\usepackage{cleveref}
\usepackage{csquotes}
\usepackage{mathtools}
\usepackage{multirow}
\usepackage[colorinlistoftodos]{todonotes}

\DeclareMathOperator*{\argmax}{argmax}

\newcommand{\jonnote}[1]{}
\newcommand{\joshnote}[1]{}
\newcommand{\ggnote}[1]{}
\newcommand{\junyaonote}[1]{}

\newcommand{\E}{\mathbb{E}}

\newcommand{\R}{\mathbb{R}}

\newcommand{\poly}{\mathrm{poly}}

  \newcommand{\costvector}{\bm{c}}
  \newcommand{\Cost}[1]{c_{#1}}
  \newcommand{\CostVector}{\bm{c}}

  \newcommand{\TypedForecast}[3]{F^{(#1)}_{#2, #3}}

  \newcommand{\ForecastTensor}{\bm{F}}
  
  \newcommand{\reward}[1]{r_{#1}}
  \newcommand{\rewardvector}{\bm{r}}
  \newcommand{\Reward}[1]{r_{#1}}
  \newcommand{\RewardVector}{\bm{r}}
  
  \newcommand{\PrincipalAgentProblem}{\left(\costvector, \ForecastTensor, \rewardvector\right)}

  \newcommand{\Contract}[1]{x_{#1}}
  \newcommand{\TypedContract}[2]{x^{(#1)}_{#2}}
  \newcommand{\ContractVector}{\bm{x}}
  \newcommand{\TypedContractVector}[1]{\bm{x}^{(#1)}}

  \newcommand{\ContractMatrix}{\bm{X}}
  
  \newcommand{\ContractDistributions}{\mathcal{X}}

  \newcommand{\OptimalAction}{i_*}
  \newcommand{\TypedOptimalAction}[1]{\OptimalAction^{(#1)}}

  \newcommand{\Profit}{\textsc{Profit}}
  \newcommand{\TypedProfit}[1]{\Profit^{(#1)}}
  \newcommand{\Utility}{U}

  \newcommand{\contract}[1]{x_{#1}}

\newcommand{\ODetMenu}{\textsc{Opt-DetMenu}}
\newcommand{\ORndMenu}{\textsc{Opt-RndMenu}}
\newcommand{\ORndMenuLinear}{\textsc{Opt-RndMenuLinear}}
\newcommand{\OSingle}{\textsc{Opt-Single}}
\newcommand{\OLinear}{\textsc{Opt-Linear}}

\allowdisplaybreaks

\newtheorem{theorem}{Theorem}[section]
\newtheorem{lemma}[theorem]{Lemma}
\newtheorem{corollary}[theorem]{Corollary}

\newtheorem{example}[theorem]{Example}

\title{The Power of Menus in Contract Design}
\author{Guru Guruganesh\\ Google Research\\\texttt{gurug@google.com} \and Jon Schneider\\ Google Research\\\texttt{jschnei@google.com} \and Joshua Wang\\ Google Research\\\texttt{joshuawang@google.com} \and Junyao Zhao\\ Stanford University\\\texttt{junyaoz@stanford.edu}}
\date{}

\begin{document}

\maketitle

\begin{abstract}
We study the power of menus of contracts in principal-agent problems with adverse selection (agents can be one of several types) and moral hazard (we cannot observe agent actions directly). For principal-agent problems with $T$ types and $n$ actions, we show that the best menu of contracts can obtain a factor $\Omega(\max(n, \log T))$ more utility for the principal than the best individual contract, partially resolving an open question of Guruganesh et al.~(2021). We then turn our attention to randomized menus of linear contracts, where we likewise show that randomized linear menus can be $\Omega(T)$ better than the best single linear contract. As a corollary, we show this implies an analogous gap between deterministic menus of (general) contracts and randomized menus of contracts (as introduced by Castiglioni et al.~(2022)).
\end{abstract}

\section{Introduction}
The principal-agent problem studies a setting where one party (the \textit{principal}) wishes to incentivize a second party (the \textit{agent}) to exert effort on behalf of the principal. The agent has different potential actions they could take (e.g., levels of effort they exert), each of which stochastically results in one of several outcomes. The principal cannot directly observe the agent's action (``moral hazard'') but they have preferences over the different outcomes and would like to incentivize the agent to take a favorable action for the principal. To accomplish this, the principal can offer the agent a \textit{contract} -- a mechanism describing how the principal will reward the agent contingent on the ultimate realized outcome. Principal-agent problems arise in a wide variety of different disciplines (e.g., law, insurance, employment) and the development of contract theory has proven to be an invaluable tool in the economics literature for the analysis of such problems.

Over the past few years, there has been a surge of activity in the application of computational methods to contract theory (much in the same way that computational methods have been successfully applied to mechanism design and auction theory). Many of these papers take the perspective of understanding the power of one class of contracting mechanisms as it compares to another. For example, in \citet{dutting2019simple}, the authors show that in principal-agent settings with $n$ actions, the optimal general contract obtains at most $\Omega(n)$ times as much utility as the optimal \emph{linear} contract (a much simpler subclass of contracts where the principal promises to share a fixed proportion of their utility with the agent), and that there are cases where this is tight. This methodology has been applied to many variants of the principal-agent problem (e.g. variants with combinatorially structured actions and outcomes).

One particularly interesting variant of the principal-agent problem introduces  private information (or ``adverse selection'') to the problem instance. In this setting an agent may be one of several types (unknown to the principal), and their relevant properties (e.g., the probability they induce a certain outcome by playing a specific action) may differ from type to type, and this allows us to model the uncertainty the principal may have over the agent they are contracting (and capture a much more realistic set of problems).

In the typed principal-agent problem, there are (at least two) natural classes of mechanisms to compare. One option is to, as before, offer a fixed single contract to the agent. However, with the presence of types, the principal also has the option to offer the agent a \emph{menu} of different contracts (from which the agent selects their favorite). Is it ever in the principal's interest to do this? This question was originally proposed by \citet{guruganesh2021contracts}, who showed there exist instances where it is strictly beneficial to offer a menu over a single contract. However, \citet{guruganesh2021contracts} were only able to show that the gap in power between menus and single contracts lies somewhere between $\Omega(n\log T)$ (for a setting with $n$ actions and $T$ types) and $2$. This gap (and even the question of whether it is super-constant) remains open to this day. 

Making things even more interesting, \citet{castiglioni2022designing} recently showed that there is a third, even more powerful class of mechanisms where the principal offers the agent a menu of \emph{randomized contracts} (i.e., the agent selects a distribution over contracts and receives a contract randomly drawn from this distribution). Although much is known about the computational properties of such mechanisms\footnote{Intriguingly, while single contracts and ordinary menus are computationally hard to optimize in the typed setting, there exist efficient algorithms to find the best menu of randomized contracts.}, essentially nothing is currently known about the power of such mechanisms -- either in comparison to ordinary menus or to individual contracts. 

\subsection{Our results}

In this paper we pin down more closely the power of menus in contract design. We begin by significantly closing the gap present between single contracts and menus of (deterministic) contracts present in \citet{guruganesh2021contracts}. In particular, \citet{guruganesh2021contracts} prove that in any principal-agent problem with $n$ actions and $T$ types, the profit of the best menu of contracts is at most $O(n\log T)$ larger than the profit of the best individual contract. However, the only lower bound they provide (and the only lower bound known to date) is that there are menus of contracts that outperform individual contracts by a constant factor. In our first two theorems, we show that both the dependence on $n$ and $T$ in the upper bound are asymptotically tight.

\begin{theorem}[Restatement of~\Cref{thm:Omega_n_gap}]
\label{thm:informal_omega_n_gap}
There exists a principal-agent problem with $3$ types and $n$ actions where the profit of the optimal menu of  (deterministic) contracts is at least $\Omega(n)$ larger than the profit of the optimal single contract. 
\end{theorem}

\begin{theorem}[Restatement of~\Cref{thm:Omega_log_T_gap}]
\label{thm:informal_omega_log_t_gap}
There exists a principal-agent problem with $T$ types and $2$ actions where the profit of the optimal menu of  (deterministic) contracts is at least $\Omega(\log T)$ larger than the profit of the optimal single contract.
\end{theorem}

We remark that both the optimal menu of deterministic contracts and the optimal single contract are known to be computationally intractable and thus can be tricky to work with. Previous approaches to show such gaps (such as~\citet{guruganesh2021contracts}) worked via adding gadgets to equate the value of one class to the value of welfare and the value of the other side to the value of linear contracts. Our proof introduces some novel elements that help bar single contracts from doing as well as menus, and does not attempt to reduce to welfare versus linear contracts.

We then switch our attention to the setting of linear contracts. Linear contracts are an important subclass of contracts in which the principal gives the agent a fixed percentage of the principal's total reward. Linear contracts are commonly used in practice, and much of the recent algorithmic work in contract theory has been focused on better understanding linear contracts in a variety of settings.

At first, it may not seem that menus of linear contracts are particularly useful -- given a menu of different fixed percentages, it is always in an agent's interest to choose the highest fixed percentage (and thus the principal could simply have offered that single linear contract). However, this assumes that we are offering the agent a menu of deterministic linear contracts. This is not the case for menus of \textit{randomized} linear contracts (in the same sense as the menus of randomized general contracts of \citet{castiglioni2022designing}) which \textit{do} have the potential to outperform a single linear contract. The interaction between the principal and the agent for a menu of randomized linear contracts proceeds as follows:

\begin{enumerate}
    \item The principal begins by constructing a menu of randomized linear contracts. Each menu item is a distribution over transfer coefficients $\alpha \ge 0$, where a specific $\alpha$ represents the linear contract where the principal offers an $\alpha$ fraction of their reward to the agent.
    \item The agent (with a randomly drawn type) selects a single distribution $\mathcal{D}$ from this menu.
    \item The principal then samples a linear contract $\alpha$ from $\mathcal{D}$ and reveals it to the agent.
    \item The agent observes the linear contract $\alpha$ and then takes an action in response. An outcome occurs, the principal receives some reward $R$, gives $\alpha R$ to the agent, and keeps $(1-\alpha)R$. 
\end{enumerate}

It turns out that randomized menus of linear contracts are in general more powerful than individual linear contracts. In particular, we show that there exist principal-agent problems where the principal can obtain a strictly larger profit by offering a menu of randomized linear contracts than by offering a single linear contract (see Example~\ref{ex:rlc-v-linear}).

We then demonstrate how to find the optimal randomized menu of linear contracts by solving a single linear program with $\poly(n, T)$ variables and constraints (\Cref{thm:polytime_menu_rlc}). In the course of doing so, we also prove the following facts about the structure of randomized menus of linear contracts, which may be of independent interest. First, we prove that without loss of generality, all randomized linear contracts in the optimal menu have identical support, and this support has size at most $O(nT)$ (see~\Cref{lem:brkpoint}). Secondly (and much more counterintuitively), this optimal support sometimes contains contracts where the principal is \emph{guaranteed to lose money}; i.e., sometimes it is necessary to include linear contracts with $\alpha > 1$ in the menu to achieve the optimal profit for the principal. We give one example of this in~\Cref{lem:unbounded_rlc_example}. 

Finally, we attempt to characterize the gap in power between menus of randomized linear contracts and individual linear contracts. 
An upper bound immediately results from the gap of \citet{guruganesh2021contracts} between a single linear contract and the maximum achievable welfare in a principal agent problem. We prove a tight lower bound for one specific regime of $n$ and $T$ (specifically, where $n$ and $T$ have roughly the same magnitude). 

\begin{theorem}[Restatement of~\Cref{thm:rlc_gap}]\label{thm:intro_rlc}
There exists a principal-agent problem with $T$ types and $O(T)$ actions where the profit of the optimal menu of randomized linear contracts is at least $\Omega(T)$ as large as the profit of the optimal individual linear contract.
\end{theorem}

Combining Theorem \ref{thm:intro_rlc} with the facts that (i) there are some principal-agent problems where any contract can be equivalently written as a linear contract and (ii) deterministic menus of linear contracts are equal in power to individual linear contracts, we immediately obtain an analogous gap between menus of (general) deterministic contracts and menus of (general) randomized contracts. 

\begin{corollary}[Restatement of Corollary \ref{cor:det_vs_rnd}]
There exists a principal-agent problem with $T$ types and $O(T)$ actions where the profit of the optimal menu of randomized contracts is at least $\Omega(T)$ as large as the profit of the optimal menu of deterministic contracts.
\end{corollary}

It is natural to ask whether it is possible to strengthen Theorem \ref{thm:intro_rlc} to understand how this gap scales individually with $n$ and $T$. That is, how large is the gap between these two types of mechanisms 
in the regime where there are a constant number of types but the number of actions for each type grows large (analogous to Theorem \ref{thm:informal_omega_n_gap}) or in the regime where there is a fixed number of actions per type, but the number of different types of agents grows large (analogous to Theorem \ref{thm:informal_omega_log_t_gap})? We do not resolve this question in this paper -- however, we make the following related observation: if every type has exactly one non-null action (e.g., every type decides between working and not working), the ratio in the revenue achieved between menus of randomized linear contracts and a single linear contract is at most a (universal) constant. 

\begin{theorem}[Restatement of~\Cref{thm:Omega_1_lower}]
In any principal-agent problem where each type of agent has one null action (guaranteeing the principal zero utility) and one non-null action, the profit of the best single linear contract is at least a constant fraction of the profit of the best menu of randomized linear contracts.
\end{theorem}

\subsection{Related work}

The principal-agent problem has a long history in economics, stemming back to the foundational papers of \citet{Holmstrom79} and \citet{GrossmanH83}. Our paper fits into a recent line of work applying techniques from algorithmic mechanism design to contract theory \citep{alon2021contracts, castiglioni2021bayesian, castiglioni2022designing, dutting2019simple, DuttingRT20, DuttingEFK21, guruganesh2021contracts,cohen2022learning}. 

Of these, the most closely relevant are a handful of recent works that consider the intersection of contract theory with private information (``contracts with types''). We build most directly off of \citet{guruganesh2021contracts}, which introduced the problem of determining the gap in power between single contracts and menus of contracts, and \citet{castiglioni2022designing}, which introduced the concept of menus of randomized contracts and demonstrated they can be computed efficiently. Other works in this space include \citet{alon2021contracts}, which studies a one-dimensional subclass of the typed principal-agent problem where the type simply scales the cost, and \citet{castiglioni2021bayesian}, which studies the computational hardness and power of optimal single contracts in the typed principal-agent setting. Also of note is the recent paper of \citet{gan2022optimal}, which provides a revelation principle for general principal-agent problems which implies (among other things) that menus of randomized contracts are the most general possible mechanisms for our setting.

We briefly mention that the theme of more powerful mechanisms (menus / randomized menus in particular) outperforming simpler mechanisms is a common theme throughout many other parts of algorithmic game theory (for example, the classic work of \citet{briest2010pricing}, which proves that randomization is essential to get approximately revenue-optimal mechanisms for selling multiple goods). These settings generally lack the property of moral hazard (i.e., hidden actions) inherent to the contract design setting, and hence these existing results and techniques do not seem to transfer over to the principal-agent problem. 







\section{Preliminaries}\label{sec:prelims}



We introduce the typed principal-agent problem following the notation of \citet{guruganesh2021contracts}. In this problem, there are $T$ different types of agent (drawn from a known prior distribution), each of which has $n$ actions which result in one of $m$ possible outcomes. When the agent of type $t \in [T]$ takes action $i \in [n]$ they incur a cost $\Cost{i} \geq 0$ for the agent (regardless of their type) and cause outcome $j \in [m]$ to occur with probability $\TypedForecast{t}{i}{j}$. When outcome $j$ occurs, the principal receives a reward $\Reward{j}$. Notably, the principal can only observe the eventual outcome $j$ and cannot observe the type $t$ of the agent or the action $i$ taken by the agent (\textit{the hidden-action model with private information}). We write $\CostVector$, $\ForecastTensor$, and $\RewardVector$ to denote the collections of costs, outcome probabilities, and rewards respectively. Together, the tuple $\PrincipalAgentProblem$ completely specifies a typed principal-agent problem instance.


We further always make the assumption that there exists a null action with zero cost ($\Cost{0} = 0$) which deterministically results in a null outcome with zero reward for the principal ($\Reward{0} = 0$; other actions may also lead to the null outcome with some probability). In this way we model the possibility for an agent to opt out of participating entirely. 

We consider three classes of contracting mechanisms for the principal (and their linear variants, which we define later): (i) single contracts, (ii) menus of deterministic contracts, and (iii) menus of randomized contracts. We discuss these in order.

\paragraph{Single contracts}
A contract is simply an $m$-dimensional vector $\ContractVector$ with non-negative entries. For each $j \in [m]$, $\Contract{j}$ specifies the amount the principal promises to transfer to the agent in the case that outcome $j$ occurs. In response to a contract $\ContractVector$, the agent of type $t$ chooses the action $\TypedOptimalAction{t}(\ContractVector)$ that optimizes their utility, namely

\begin{align*}
  \TypedOptimalAction{t}(\ContractVector) \in \argmax_i \left( \sum_{j=1}^m \TypedForecast{t}{i}{j} \Contract{j} \right) - \Cost{i} \text{.}
\end{align*}

We call the principal's net expected utility from offering contract $\ContractVector$ their \emph{profit} from offering contract $\ContractVector$. The principal's profit from an agent of type $t$ is given by

\begin{align*}
  \TypedProfit{t}(\CostVector, \ForecastTensor, \RewardVector, \ContractVector)
    \triangleq \sum_{j=1}^m \TypedForecast{t}{\TypedOptimalAction{t}(\ContractVector)}{j} (\reward{j} - \contract{j})
\end{align*}
and their overall expected profit is given by 
\begin{align*}
  \Profit(\CostVector, \ForecastTensor, \RewardVector, \ContractVector)
    \triangleq \E_t \left[ \TypedProfit{t}(\CostVector, \ForecastTensor, \RewardVector, \ContractVector) \right] \text{.}
\end{align*}

We are interested in the largest profit obtainable for the principal via a single contract. We write $\OSingle \PrincipalAgentProblem$ to denote this maximal profit for the principal agent problem $\PrincipalAgentProblem$.

\paragraph{Menus of deterministic contracts}

In the presence of types, the principal can sometimes obtain additional profit by offering the agent a menu of individual contracts (from which the agent picks the contract that maximizes their expected utility). By the revelation principle, it suffices to specify a single contract for each type, and we can write a general menu of deterministic contracts in the form $\ContractMatrix = \left(\TypedContractVector{1}, \ldots, \TypedContractVector{T}\right)$, where an agent who reports their type to be $t$ receives the contract $\TypedContractVector{t}$. Such a menu is ``incentive-compatible'' (IC) if no type $t$ has an incentive to misreport their type as a different type $t' \ne t$:
\begin{align}
  \label{ineq:ic}
  \forall t, t' \in [T] \quad \max_i \left( \sum_{j=1}^m \TypedForecast{t}{i}{j} \TypedContract{t}{j} \right) - c_i
                          \ge \max_i \left( \sum_{j=1}^m \TypedForecast{t}{i}{j} \TypedContract{t'}{j} \right) - c_i
  \text{.}
\end{align}

Unless otherwise stated, we restrict our attention to incentive-compatible menus of deterministic contracts. As with single contracts, we write $\Profit(\CostVector, \ForecastTensor, \RewardVector, \ContractMatrix)$ to denote the expected profit for the principal from offering the menu of deterministic contracts $\ContractMatrix$, and we write $\ODetMenu{}\PrincipalAgentProblem$ to be the profit of the best menu of deterministic contracts. 

\paragraph{Menus of randomized contracts}

Deterministic menus can be generalized further by offering distributions over contracts instead of individual contracts. This class of mechanisms was proposed by \citet{castiglioni2022designing}, and works as follows. The principal offers the agent a menu of distributions over contracts (from which the agent picks the distribution that maximizes their expected utility, with randomness in both which contract gets chosen and which outcome occurs). By the revelation principle, it suffices to specify a distribution for each type, and we can write a general menu of randomized contracts in the form $\ContractDistributions = \left(\mathcal{D}^{(1)}, \mathcal{D}^{(2)}, \ldots, \mathcal{D}^{(T)}\right)$, where an agent who reports their type to be $t$ receives the distribution $\mathcal{D}^{(t)}$. The principal then draws a contract from $\mathcal{D}^{(t)}$ and relays it to the agent, who takes an action in response. Such a menu is ``incentive-compatible'' (IC) if no type $t$ has an incentive to misreport their type as a different type $t' \ne t$:
\begin{align*}
  \forall t, t' \in [T] \quad
    \E_{x \sim \mathcal{D}^{(t)}} \left[ \max_i \left( \sum_{j=1}^m \TypedForecast{t}{i}{j} \Contract{j} \right) - c_i \right]
    \ge \E_{x \sim \mathcal{D}^{(t')}} \left[ \max_i \left( \sum_{j=1}^m \TypedForecast{t}{i}{j} \Contract{j} \right) - c_i \right]
  \text{.}
\end{align*}

We restrict our attention to incentive-compatible menus of randomized contracts. As before, we write $\Profit(\CostVector, \ForecastTensor, \RewardVector, \ContractDistributions)$ to denote the expected profit for the principal from offering the menu of randomized contracts $\ContractDistributions$, and we write $\ORndMenu{}\PrincipalAgentProblem$ to be the profit of the best menu of randomized contracts.

\paragraph{Linear contracts}
We now discuss how to restrict the previous three classes to linear contracts. Contracts with the property that $\Contract{j} = \alpha \cdot \Reward{j}$ for all $j \in [m]$ and some $\alpha \geq 0$ are called \emph{linear contracts}.

The most straightforward class to restrict is that of single contracts, where all we do is require that the chosen contract $\ContractVector$ is linear. We write $\OLinear{}\PrincipalAgentProblem$ to be the profit of the best (single) linear contract.

The next class to restrict is menus of deterministic contracts, where we require that each individual contract offered by the menu is linear. Specifically, for $\ContractMatrix = \left( \TypedContractVector{1}, \TypedContractVector{2}, \ldots, \TypedContractVector{T}\right)$, we would like each $\TypedContractVector{t}$ to be linear. Interestingly, this turns out to be no more powerful than (single) linear contracts, because agents might as well pick the offered linear contract with the highest $\alpha$. As a result, $\OLinear{}\PrincipalAgentProblem$ also designates the profit of the best menu of deterministic linear contracts.

The final class to restrict is menus of randomized contracts, where we require that each distribution offered by the menu only contains linear contracts in its support. Specifically, for $\Gamma = \left( \mathcal{\gamma}^{(1)}, \mathcal{\gamma}^{(2)}, \ldots, \mathcal{\gamma}^{(T)} \right)$ we would like the support of $\mathcal{\gamma}^{(1)}$ to be linear contracts. We will show this is more powerful than (single) linear contracts. We write $\ORndMenuLinear{}\PrincipalAgentProblem$ to be the profit of the best menu of randomized linear contracts. These menus will be our main object of study in Section \ref{sec:menu_rlc}; as we will see, unlike menus of deterministic linear contracts, these are occasionally more powerful than single linear contracts. 

\begin{figure}
\centering
\begin{tikzpicture}[%
  auto,
  scale=1.0,
  ladderrung/.style={
    rectangle,
    draw=black,
    inner sep=4pt,
    align=center,
  },
  ]
  \node[ladderrung] (rndmenu) at (0, 4) {Menus of\\Randomized Contracts};
  \node[ladderrung] (detmenu) at (-2.5, 2.5) {Menus of\\Deterministic Contracts};
  \node[ladderrung] (rndmenulinear) at (2.5, 1.75) {Menus of Randomized\\Linear Contracts};
  \node[ladderrung] (single) at (-2.5, 1) {Single Contracts};
  \node[ladderrung] (linear) at (0, 0) {Linear Contracts};

  \draw (rndmenu) -- (detmenu) -- (single) -- (linear);
  \draw (rndmenu) -- (rndmenulinear) -- (linear);
\end{tikzpicture}
\caption{Hierarchy of contracting mechanisms for typed contract problems. Edges link more general, powerful classes (higher) to more restricted, weaker classes (lower). The left (menus of deterministic contracts and single contracts) and right (menus of randomized linear contracts) sides of the diagram are incomparable; for some instances a left class is more powerful and for other instances a right class is more powerful.}
\label{fig:ladder}
\end{figure}
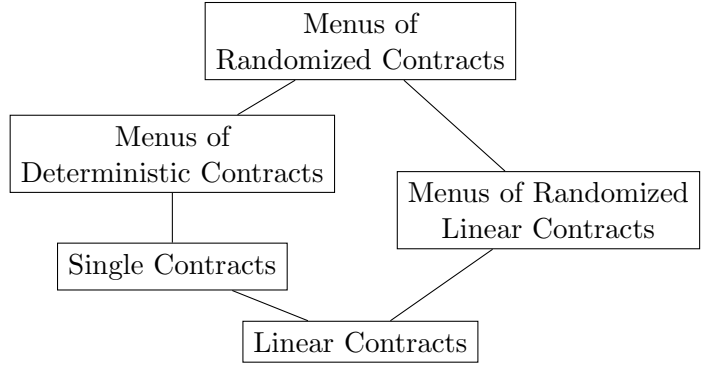
This completes our overview of five classes of contracting mechanisms for the principal. Our five classes of interest are depicted in Figure~\ref{fig:ladder}, which also illustrates which classes contain each other.

\subsection{Known upper bounds}

\citet{guruganesh2021contracts} proved an upper bound of $O(n \log T)$ between a single linear contract and the maximum achievable welfare in a typed principal agent problem. Every class mentioned above is at least as powerful as single linear contracts and at most as powerful as extracting all welfare, so this $O(n \log T)$ bound applies to any pair of classes as well. However, this upper bound hinges on two key assumptions: (i) all types have the same set of costs and (ii) the distribution over types is uniform. Removing either of these assumptions breaks the proof, leaving only a $O(nT)$ bound. In this paper, we will work in the more general setting where the costs may be vary with types and the distribution is nonuniform. However, all proofs can be adapted to work in the more restricted setting above.

\section{Gaps between menus and single contracts}
In this section, we show that deterministic menus of contracts can be much more powerful than a single contract:
\begin{itemize}
    \item When the number of types $T$ is a constant, there is a principal-agent problem instance for which optimal deterministic menu of contracts can extract $\Omega(n)$ times as much profit as the optimal single contract.
    \item When the number of actions $n$ is a constant, there is an instance where optimal deterministic menu of contracts can extract $\Omega(\log T)$ times as much profit as the optimal single contract.
\end{itemize}
\subsection{\texorpdfstring{$\Omega(n)$}{O(n)}-gap with three types}
\begin{theorem}\label{thm:Omega_n_gap}
    There is a principal-agent problem $\PrincipalAgentProblem$ with $T=3$ agent types and $n$ actions, for which optimal deterministic menu can extract $\Omega(n)$ times as much profit as optimal single contract can: 
    \[\ODetMenu \PrincipalAgentProblem\ge \Omega(n) \cdot \OSingle \PrincipalAgentProblem.\]
\end{theorem}
We provide here the $\Omega(n)$-gap instance and explain the high-level ideas behind its construction and analysis. We defer the full proof of Theorem~\ref{thm:Omega_n_gap} (which is rather technical) to \Cref{subsec:apx-omega-n-gap}.

\paragraph*{Construction of $\Omega(n)$-gap instance}

We begin by providing a formal description of the instance, which is also recorded in Table \ref{tab:omega-n}. There are three agent types. There are four outcomes, and their rewards are $r_1=1$, $r_2=0$, $r_3=1$ and $r_4=0$. We will assume we have $2n+1$ non-null actions (assuming $n\ge 12$ without loss of generality), with costs given by $c_i=\frac{n^i-i}{n^{n+1}}$ for $i\in[n]$, $c_{n+i}=\frac{4(n^i-i)}{n^{n+1}}$ for $i\in[n]$, and $c_{2n+1}=0$.

We now construct the forecast tensor $\ForecastTensor$. First, we let $F^{(1)}_{2n+1,j}=0$, $F^{(1)}_{i+n,j}=0$ and $F^{(2)}_{i,j}=0$ for all $i\in[n]$ and $j\in\{1,2,3\}$, and we let $F^{(3)}_{i,j}=0$ for all $i\in[2n+1]$ and $j\in\{1,2,3\}$. That is, outcome $4$ always occurs when the type $1$ agent plays the last $n+1$ actions, when the type $2$ agent plays the first $n$ actions, and when the type $3$ agent plays any action. 

For the type $1$ agent, we let $F^{(1)}_{i,1}=n^{i-1-n}$ for all $i\in[n]$, and we let $F^{(1)}_{i,2}=\frac{n^i-i}{n^n-n}\cdot F^{(1)}_{n,2}$ for all $i\in[n]$, and moreover, we let $F^{(1)}_{i,3}=0$ for all $i\in[n]$. For the type $2$ agent, we let $F^{(2)}_{n+i,3}=4n^{i-1-n}$ for all $i\in[n]$, let $F^{(2)}_{n+i,2}=\frac{4}{n^{n-1}-1}\cdot F^{(1)}_{n,2}$ for all $i\in[n]$, and let $F^{(2)}_{n+i,1}=0$ for all $i\in[n]$. Furthermore, we let $F^{(2)}_{2n+1,1}=F^{(2)}_{2n+1,3}=0$, and $F^{(2)}_{2n+1,2}=\frac{4}{n^{n-1}-1}\cdot F^{(1)}_{n,2}$. The specific value of $F^{(1)}_{n,2}$ will not play any role in the analysis, and we can for example let $F^{(1)}_{n,2}=\frac{1}{3}$ for completeness (in Table \ref{tab:omega-n} we denote this value by $\gamma$). 

When $n\ge12$, it is easy to check that $F^{(t)}_{i,j}\le \frac{1}{3}$ for all $i\in[2n+1]$, $j\in\{1,2,3\}$ and $t\in\{1,2,3\}$. Therefore, we can let $F^{(t)}_{i,4}=1-\sum_{j\in\{1,2,3\}}F^{(t)}_{i,j}$ for all $i\in[2n+1]$ and $t\in\{1,2,3,4\}$ such that the forecast tensor $\ForecastTensor$ is well-defined. Finally, we specify the probability of agent types: $\Pr[\textrm{agent type}=1]=\Pr[\textrm{agent type}=2]=\frac{1}{2n^{2n}}$ and $\Pr[\textrm{agent type}=3]=1-\frac{1}{n^{2n}}$.


\begin{table}
\centering
\begin{tabular}{cccccc}
  \toprule
  \multirow{2}{*}{Type $1$} & \textbf{Outcome} $1$ & \textbf{Outcome} $2$ & \textbf{Outcome} $3$ & \textbf{Outcome} $4$ \\
                                    & Reward $1$           &  Reward $0$          & Reward $1$           & Reward $0$             \\ \midrule
  \textbf{Action} $i \in [n]$ & \multirow{2}{*}{$\frac{1}{n^{n-(i-1)}}$} & \multirow{2}{*}{$\frac{n^i - i}{n^n - n}\gamma$} & \multirow{2}{*}{$0$} & \multirow{2}{*}{$1 - \sum_{j=1}^{3}F^{(1)}_{i,j}$} \\ Cost $(n^i - i)/n^{n+1}$ \\
  \textbf{Action} $n + i, i\in [n]$ & \multirow{2}{*}{0} & \multirow{2}{*}{0} & \multirow{2}{*}{0} & \multirow{2}{*}{$1$} \\ Cost $4(n^{i} - i)/n^{n+1}$ \\
  \textbf{Action} $2n+1$ & \multirow{2}{*}{$0$} & \multirow{2}{*}{$0$} & \multirow{2}{*}{$0$} & \multirow{2}{*}{$1$} \\ Cost $0$ \\ \midrule
    \multirow{2}{*}{Type $2$} & \textbf{Outcome} $1$ & \textbf{Outcome} $2$ & \textbf{Outcome} $3$ & \textbf{Outcome} $4$ \\
                                    & Reward $1$           &  Reward $0$          & Reward $1$           & Reward $0$             \\ \midrule
  \textbf{Action} $i \in [n]$ & \multirow{2}{*}{$0$} & \multirow{2}{*}{$0$} & \multirow{2}{*}{$0$} & \multirow{2}{*}{$1$} \\ Cost $(n^i - i)/n^{n+1}$ \\
  \textbf{Action} $n + i, i\in [n]$ & \multirow{2}{*}{$0$} & \multirow{2}{*}{$\frac{4}{n^{n-1} - 1}\gamma$} & \multirow{2}{*}{$\frac{4}{n^{n-(i-1)}}$} & \multirow{2}{*}{$1 - \sum_{j=1}^{3}F^{(2)}_{n+i,j}$} \\ Cost $4(n^{i} - i)/n^{n+1}$ \\
  \textbf{Action} $2n+1$ & \multirow{2}{*}{$0$} & \multirow{2}{*}{$\frac{4}{n^{n-1} - 1}\gamma$} & \multirow{2}{*}{$0$} & \multirow{2}{*}{$1 - \sum_{j=1}^{3}F^{(2)}_{2n+1,j}$} \\ Cost $0$ \\ \midrule
  \multirow{2}{*}{Type $3$} & \textbf{Outcome} $1$ & \textbf{Outcome} $2$ & \textbf{Outcome} $3$ & \textbf{Outcome} $4$ \\
                                    & Reward $1$           &  Reward $0$          & Reward $1$           & Reward $0$             \\ \midrule
  \textbf{Action} $i \in [2n+1]$ & \multirow{2}{*}{$0$} & \multirow{2}{*}{$0$} & \multirow{2}{*}{$0$} & \multirow{2}{*}{$1$} \\ Cost $c_i$ \\
  \bottomrule
\end{tabular}
\caption{The $\Omega(n)$ gap instance of Theorem \ref{thm:Omega_n_gap}. Here we can choose $\gamma$ to be any (sufficiently small) positive constant; setting $\gamma = 1/3$ ensures that this forms a valid typed principal-agent problem for all $n \geq 12$. Agent types occur with probabilities $1/(2n)^{2n}$, $1/(2n)^{2n}$, and $1 - 2/(2n)^{2n}$ respectively. Outcome 4 can be thought of as a null action, and type 3 exists solely to incentivize the principal to transfer no reward on Outcome 4 (Lemma \ref{lemma:x4_upper_bound_Omega_n_gap}). The best single contract achieves a profit of $O(n^{-(n+1)})$, but the optimal menu (which contains two contracts, one transferring only on outcome 2, and one transferring only on outcome 3) achieves a profit of $\Omega(n^{-n})$.}
\label{tab:omega-n}
\end{table}
\paragraph{High-level idea} First of all, we include the zero-reward outcome $4$ merely to satisfy the technical constraint that the probabilities over outcomes sum up to $1$ for any type and action. Ideally, we would like to ignore the existence of outcome $4$, but in principle, a contract can specify a strictly positive payment for outcome $4$. Therefore, we introduce type $3$, which has much higher probability than the other two types and always results in outcome $4$ to make sure the payment for outcome $4$ is negligible (Lemma~\ref{lemma:x4_upper_bound_Omega_n_gap}) and hence does not interfere with the analysis. This is the only role of type $3$, and thus we will ignore outcome $4$ and type $3$ for the remainder of this discussion.

Essentially, the above $\Omega(n)$-gap instance is constructed such that the principal can only get low profit from type $2$ agent regardless of the contract (Lemma~\ref{lemma:agent_2_profit_upper_bound_Omega_n_gap}), and moreover, the only way to get high profit from type $1$ agent is using a contract that associates a sufficiently low payment with outcome $1$ and a sufficiently high payment with outcome $2$ (Lemma~\ref{lemma:agent_1_low_profit_for_large_x1_Omega_n_gap} and Lemma~\ref{lemma:lower_bound_of_x2_Omega_n_gap}), i.e., the contract must incentivize type $1$ agent to play a high-profit action through a high payment for outcome $2$. However, notice that in the $\Omega(n)$-gap instance, the type $2$ agent also has significant probability for outcome $2$ for every non-null action. Thus, if the contract specifies a high payment for outcome $2$, it makes a high payment to the type $2$ agent for almost nothing in return and hence achieves negative profit from the type $2$ agent. Overall, this lets us show that the combined profit from these two types of agents is low (the soundness part of the proof of Theorem~\ref{thm:Omega_n_gap}). This shows that any single contract has low profit.

On the other hand, notice that in this instance, outcome 3 has a strictly positive reward, and only type $2$ agent can possibly cause outcome 3 to happen. Therefore, if we use a menu of two contracts, we can use the first contract to incentivize type $1$ agent through a high payment for outcome 2 and use the second contract to lure type $2$ agent away from the first contract through a high payment for outcome 3. We can show that the second contract can simultaneously (i) be as attractive as the first contract for type $2$ agent and (ii) make zero profit (which is non-negative) from type $2$ agent (the completeness part of the proof of Theorem~\ref{thm:Omega_n_gap}). Thus, overall the menu achieves high profit (in particular, at least $\Omega(n)$ times the profit of the best single contract).

\subsection{\texorpdfstring{$\Omega(\log T)$}{O(log T)}-gap with two non-null actions}
\begin{theorem}\label{thm:Omega_log_T_gap}
    There is a principal-agent problem $\PrincipalAgentProblem$ with $T$ agent types and $n=2$ non-null actions, for which optimal deterministic menu can extract $\Omega(\log T)$ times as much profit as optimal single contract can: 
    \[\ODetMenu \PrincipalAgentProblem\ge \Omega(\log T) \cdot \OSingle \PrincipalAgentProblem.\]
\end{theorem}
We first give the construction of the $\Omega(\log T)$-gap instances and explain the high-level idea. Then, we formally implement our idea through multiple lemmas and prove Theorem~\ref{thm:Omega_log_T_gap}.

\begin{table}
\centering
\begin{tabular}{ccccc}
  \toprule
  Type $(k-1)2^N + \ell$        & \textbf{Outcome} $1$ & \textbf{Outcome} $j \in \{2, 3, \ldots, T-1\}$ & \textbf{Outcome} $T$ \\
  ($k \in [N], \ell \in [2^N]$) & Reward $1$           &  Reward $0$                                    & Reward $0$           \\ \midrule
  \textbf{Action} $1$ & \multirow{2}{*}{$\frac{1}{2T(1 - 2^{-k})}$} & \multirow{2}{*}{$\frac{1}{2T(1 - 2^{-k})} \text{ if } j = (k-1)2^N + \ell$} & \multirow{2}{*}{$1 - \frac{1}{T(1 - 2^{-k})}$} \\ Cost $\frac{1}{2T}$ \\
  \textbf{Action} $2$ & \multirow{2}{*}{0} & $\frac{1}{2T(1 - 2^{-k})2^{k+1}} \text{ if } j \ne (k-1)2^N + \ell$ & \multirow{2}{*}{$1 - \frac{T-3}{2T(1 - 2^{-k})2^{k+1}}$} \\ Cost $0$ & & 0 \text{ otherwise} \\ \midrule
  \multirow{2}{*}{Type $N \cdot 2^N + 1$} & \textbf{Outcome} $1$ & \textbf{Outcome} $j \in \{2, 3, \ldots, T-1\}$ & \textbf{Outcome} $T$ \\
                                          & Reward $1$           &  Reward $0$                                    & Reward $0$           \\ \midrule
  \textbf{Action} $1$ & \multirow{2}{*}{$0$} & \multirow{2}{*}{$0$} & \multirow{2}{*}{$1$} \\ Cost $\frac{1}{2T}$ \\
  \textbf{Action} $2$ & \multirow{2}{*}{$0$} & \multirow{2}{*}{$0$} & \multirow{2}{*}{$1$} \\ Cost $0$ \\
  \bottomrule
\end{tabular}
\caption{The $\Omega(\log T)$ gap instance of Theorem \ref{thm:Omega_log_T_gap}. Agent type $(k - 1)2^N + \ell$ occurs with probability $2^{k-1} / 16^N$ for $k \in [N]$ and $\ell \in [2^N]$, and the final agent type $N \cdot 2^N + 1$ occurs with the remaining probability, which is the vast majority. This final agent type serves to tax the final outcome so it cannot be used to a meaningful extent. A single contract has difficulty using Outcome 1 alone to incentivize action 1, because of the varied breakpoints, and if the single contract uses too many other outcomes instead to incentivize action 1, most types of agents would prefer playing action 2.
}
\label{tab:omega-log-t-table}
\end{table}
\paragraph{Construction of $\Omega(\log T)$-gap instances}
We begin by providing a formal description of the instance, which is also recorded in Table \ref{tab:omega-log-t-table}. For convenience we let $T=N\cdot 2^{N}+1$ for $N\ge3$ (and we will prove an $\Omega(N)$-gap). There are two non-null actions and $T$ outcomes. Action $1$ has cost $c_1=\frac{1}{2T}$, and action $2$ has cost $c_2=0$. Outcome $1$ has reward $r_1=1$, and any other outcome $j\in\{2,3,\dots,T+1\}$ has reward $r_j=0$.

Now we construct the forecast tensor $\ForecastTensor$. We start with defining $\ForecastTensor$ for outcomes $[T-1]$. For all $k\in[N]$ and $\ell\in[2^N]$, first let $F^{((k-1)2^N+\ell)}_{1,1}=\frac{1}{2T(1-2^{-k})}$, and moreover for $j\in[T-1]$ such that $j\neq(k-1)2^N+\ell$ or $1$, let $F^{((k-1)2^N+\ell)}_{2,j}=\frac{F^{((k-1)2^N+\ell)}_{1,1}}{2^{k+1}}$. Then, let $F^{(t)}_{1,t}=F^{(t)}_{1,1}$ for all $t\in[T-1]$. All the other $F^{(t)}_{i,j}$ for $t\in[T]$, $i\in\{1,2\}$ and $j\in[T-1]$ that have not been defined are zero.

It is easy to check that $F^{(t)}_{i,j}\le\frac{1}{T}$ for all $t\in[T]$, $i\in\{1,2\}$ and $j\in[T-1]$. Thus, we can let $F^{(t)}_{i,T}=1-\sum_{j\in[T-1]}F^{(t)}_{i,j}$ for all $t\in[T]$ and $i\in\{1,2\}$, such that the forecast tensor $\ForecastTensor$ is well-defined.

Finally, we specify the probability of agent types: $\Pr[\textrm{agent type}=(k-1)2^N+\ell]=\frac{2^{k-1}}{16^N}$ for all $k\in[N]$ and $\ell\in[2^N]$, and $\Pr[\textrm{agent type}=T]=1-\frac{2^{N}-1}{8^N}$.


\paragraph{High-level idea}
Similar to the $\Omega(n)$-gap instance, here we also have an extra zero-reward outcome $T$ to make the forecast tensor well-defined, and we introduce type $T$, which has much larger probability than other types and always results in outcome $T$, in order to make sure that a contract with non-negative profit in expectation over type distribution can only make a negligible payment $x_T$ for outcome $T$ (Lemma~\ref{lemma:x_T_upper_bound_Omega_log_T_gap}), such that $x_T$ does not interfere with the analysis. Thus, let us ignore outcome $T$ and type $T$ in the following discussion.

In the above $\Omega(\log T)$-gap instance, only outcome $1$ generates strictly positive reward, and hence, only action $1$ generates strictly positive welfare for any type (because only action $1$ has strictly positive probability for outcome $1$). Moreover, the expected welfare generated by action $1$ times the probability of agent type is roughly the same (up to a constant factor) for all types: That is, for all $k\in[N]$ and $\ell\in[2^N]$,
\begin{equation}\label{eq:equal_revenue_Omega_log_T_gap}
     \Pr[\textrm{agent type}=(k-1)2^N+\ell]\cdot \left(F^{((k-1)2^N+\ell)}_{1,1}\cdot 1-c_1\right)=\frac{1}{4T\cdot16^N(1-2^{-k})}\in\left[\frac{1}{4T\cdot16^N},\frac{1}{2T\cdot16^N}\right].
\end{equation}
Therefore, to extract a constant fraction of the expected welfare of a random-type agent, the contract must incentivize many agent types $t$ to play action $1$. For the type $t$ agent, the only two outcomes that action $1$ can result into are outcomes $1$ and $t$.

Because the probability that action $1$ results in outcome $1$ is different between different types, the contract cannot incentivize different types using only payment for outcome $1$. Specifically, a small payment for outcome $1$ is not enough to cover the expected cost of action $1$ for many types, and a large payment for outcome $1$ would overpay many types and hence give up a lot of profit (Lemma~\ref{lemma:bound_lower_part_agents_Omega_log_T_gap}).

If we use a menu of contracts, we can have one contract $\ContractVector_t$ for each agent type $t$. Each contract has the same base payment for outcome $1$, which by itself is not enough to incentivize action $1$. Then, contract $\ContractVector_t$ can incentivize the type $t$ agent to play action $1$ by an additional payment through outcome $t$. Therefore, such menu can extract a constant fraction of the expected welfare of the agent (the completeness part of the proof of Theorem~\ref{thm:Omega_log_T_gap}).

However, if we adopt the above idea for the single contract rather than the menu, then, we need to make additional payments for many different outcomes in the \emph{single} contract. In the above $\Omega(\log T)$-gap instance, for the type $t$ agent, action $2$ has zero cost and has strictly positive probabilities for all but outcomes $1$ and $t$. Therefore, when faced with a single contract with additional payments for many different outcomes, the type $t$ agent would prefer playing action $2$, which generates zero profit for the principal. As a consequence, when we use a single contract, there cannot be many types of agents simultaneously preferring action $1$ (Lemma~\ref{lemma:bound_upper_part_agents_Omega_log_T_gap}), and hence, overall the expected profit is low (the soundness part of the proof of Theorem~\ref{thm:Omega_log_T_gap}).

\newtheorem{remark}[theorem]{Remark}
\begin{remark}
Under the assumption that the probability of every agent type with non-zero welfare is the same in the prior distribution of the agent type,~\citet{guruganesh2021contracts} proved that linear contracts can achieve a $\Omega(1/(n\log T))$ fraction of the profit of the optimal deterministic menu. Thus, our $\Omega(\log T)$-gap instance (where $n=3$) is tight for this setting.
\end{remark}
\begin{proof}
    Note that in the $\Omega(\log T)$-gap instance, only agent types $[T-1]$ has non-zero welfare, but their probabilities in the prior distribution are different. However, we can construct a new instance with equal probabilities of agent types from the $\Omega(\log T)$-gap instance by making copies of each agent type such that the number of copies is proportional to the probability of that agent type. That is, in the new instance, for each $t\in[T-1]$, we make $16^N\cdot\Pr[\textrm{agent type}=t]$ (it is easy to check $16^N\cdot\Pr[\textrm{agent type}=t]$ is an integer) copies of agent type $t$, and we make $16^N$ copies of agent type $T$.  In the new instance, all the types have the same probability in the new prior distribution. Clearly, the analysis of Theorem~\ref{thm:Omega_log_T_gap} still works by replacing the probability of each agent type with the number of copies of that agent type.
\end{proof}
    
\section{Menus of randomized linear contracts}
\label{sec:menu_rlc}


We now switch our attention to menus involving linear contracts. As we have already mentioned in Section \ref{sec:prelims}, simple menus of (deterministic) linear contracts are no more powerful than a single linear contract. Instead, our main object of study are \textit{menus of randomized linear contracts}, formed by specializing the menus of randomized contracts in \citet{castiglioni2022designing} to the special case of linear contracts.


Recall that a randomized linear contract is a probability distribution $\gamma$ over $\mathbb{R}_{\geq 0}$ representing all
the possible linear contracts (including linear contracts with $\alpha > 1$, where the principal will pay out more than their entire reward). A menu of randomized linear contracts is defined by a tuple $\Gamma = \left(\gamma^{(1)},\dots,\gamma^{(T)} \right)$ where each $\gamma^{(i)}$ is a distribution over the transfer coefficients.
The interaction between the principal and an agent of type $t$ proceeds as follows:
\begin{enumerate}
\item The principal publicly commits to a menu $\Gamma= \left(\gamma^{(1)},\dots,\gamma^{(T)}\right).$
\item The agent reports a type $\hat{t}$ to the principal, possibly different from the true type $t.$
\item The principal draws a transfer coefficient $\alpha \sim \gamma^{(\hat{t})}$  and communicates it to the agent.
\item The agent plays the action that maximizes his utility when offered $\alpha $ of the expected reward. 
\end{enumerate}

The very first question we might ask ourselves is: are randomized linear contracts actually more powerful than individual linear contracts. The answer (in contrast to the situation with menus of deterministic linear contracts) is \textit{yes}.

Before we present the example, we will establish some common notation for the rest of the section. Let $U_{t}(\alpha)$ be the utility obtained by type $t$ when offered a linear contract with coefficient $\alpha$. If action $i$ has cost $c_i$ and expected reward $R^{(t)}_i$ for the principal, then in fact $U_{t}(\alpha) = \max_{i} (R^{(t)}_i \alpha - c_i)$. From this we can observe that each $U_{t}$ is a convex, piece-wise linear function with at most $n$ pieces. Since we are only dealing with menus of randomized linear contracts, it suffices to specify the slopes $U'_t$ of these piecewise linear functions to completely specify the problem instance. 

\begin{example}[Menus of Randomized Linear Contracts Differ From Linear Contracts]
\label{ex:rlc-v-linear}
  Consider the following typed principal-agent problem with two types (and the underlying distribution is uniform), where $\epsilon > 0$ is a small constant. Let the slopes of their utility functions be defined as follows:
  \begin{align*}
    &U_1'(\alpha) = 
    \begin{cases}
    0  &\text{ if } \alpha \in [0,1 - \epsilon) \\
    \frac{1}{2\epsilon} - 1  &\text{ if } \alpha \in [1 - \epsilon, 1]
    \end{cases}
    \qquad 
    &U_2'(\alpha) = 
    \begin{cases}
    \frac{1}{2} &\text{ if } \alpha \in [0,\frac12) \\
    1 &\text{ if } \alpha \in [\frac12,1]
    \end{cases}
\end{align*}

This problem is designed so that any linear contract where $\alpha$ is a breakpoint (one of $\{0, 1/2, 1-\epsilon\}$) yields a total utility of $1/2$; linear contracts at non-breakpoints only get less utility. The following menu achieves slightly more than $1/2$ utility:
\begin{align*}
  \gamma^{(1)}(\alpha) = \begin{cases}
    2/3 & \text{ if } \alpha = 0 \\
    1/3 & \text{ if } \alpha = 1 - \epsilon \\
    0          & \text { otherwise}
  \end{cases}
  \qquad
  \gamma^{(2)}(\alpha) = \begin{cases}
    1 & \text{ if } \alpha = \frac12 \\
    0 & \text { otherwise}
  \end{cases}
\end{align*}
Type one is indifferent between the two menu items because it has zero utility at all breakpoints. From type two's point of view, its menu item has utility $1/4$ and the other menu item has utility $\frac13 (\frac14 + \frac12 - \epsilon) < 1/4$. The first menu item extracts a profit of $\frac16 - \frac13 \epsilon$ from the first type. The second menu item extracts a profit of $\frac12$ from the second type. The total profit is $\frac23 - \frac13 \epsilon$, which is better for $\epsilon < \frac12$. This completes the example.
\end{example}

In the remainder of this section, we study the following three aspects of randomized linear contracts. First, in Section \ref{sec:structure}, we explore various structure of randomized linear contracts: e.g., How large must the support of $\gamma^{(t)}$ be in the optimal menu of randomized linear contracts? Is it necessary for the support to contain $\alpha > 1$? Second, in Section \ref{sec:rlc_alg}, we present an efficient algorithm for computing the optimal menu of randomized linear contracts. Finally, in Sections \ref{sec:rlc_gap} and \ref{sec:rlc_ub}, we explore (and bound) the gap in power between menus of randomized linear contracts and linear contracts.

\subsection{The structure of randomized linear contracts} \label{sec:structure}

We begin by establishing some basic structural properties of the randomized linear contracts in belonging to the optimal menu. In particular, we will show that all such contracts are supported on breakpoints $\alpha$ where an agent of one type would switch from playing one action to another.


Let $p$ be the total number of distinct breakpoints over all $T$ functions $U_{t}(\alpha)$ (in particular, $p \leq nT$). We will label these breakpoints as $\alpha_1\dots  \leq \alpha_{p}$, and for convenience of notation write $\alpha_0 = 0$. The first claim we prove is that, without loss of generality, every optimal menu is supported on this set of breakpoints (including $\alpha_0$), and possibly one other point larger than $\alpha_p$.

\begin{lemma} \label{lem:brkpoint}
For any principal-agent problem with $T$ types and $n$ actions, there exists an optimal menu $\Gamma = (\gamma^{(1)},\dots,\gamma^{(T)})$ of randomized linear contracts such that, for each $t \in [T]$,

$$\mathrm{supp}\left(\gamma^{(t)}\right) \subseteq \{\alpha_0, \alpha_1, \dots, \alpha_{p}, \alpha^{t}_{p+1}\}$$

\noindent
for some choice of $\alpha^{t}_{p+1} > \alpha_{p}$ for each type $t$. In other words, all randomized contracts in this menu only place weight on points of the form $\alpha_i$ or one other (undetermined) point.


\end{lemma}

We defer the full proof to Appendix~\ref{subsec:apx-lem-brkpoint} but present a sketch below.
\begin{proof}[Proof Sketch]
Note that any mass put on a break point $\alpha \in [\alpha_j,\alpha_{j+1})$ can be 
moved to both neighboring breakpoints without affecting the constraints and only increasing the objective.  For the probability mass put on the points above the last breakpoint, we can create a new point based on its conditional expectation in a way that maintains the constraints and improves the objective. 
\end{proof}


One obvious peculiarity of the characterization in Lemma \ref{lem:brkpoint} is that it leaves open the option of placing mass above the largest breakpoint $\alpha_{p}$, without any bound on the size of $\alpha^t_{p+1}$. In fact, a priori, there is nothing preventing the optimal menu from containing contracts supported on coefficients $\alpha$ greater than $1$, even though by offering such a contract, the principal is \textit{guaranteed} to lose utility. Even barring placing mass beyond $\alpha_{p}$, it could still be the case that some breakpoints $\alpha_i$ themselves are greater than $1$; these would normally correspond to actions that are not reasonable to incentivize with a single contract (they have negative net welfare for the principal and agent combined). Is it ever useful to incentivize these actions when offering a menu of randomized linear contracts?

Interestingly, the answer is (again) \textit{yes}. We say that a randomized linear contract is \emph{bounded} if always chooses a transfer coefficients bounded above by $1$. Otherwise we say a randomized linear contract is \emph{unbounded}. In the example below, we show that a menu of unbounded randomized linear contracts can extract strictly larger revenue for the principal than a menu of bounded randomized linear contracts.

\begin{lemma}\label{lem:unbounded_rlc_example}
There exists a principal-agent problem with $T=2$ types and $n=3$ actions such that the profit obtained by the optimal menu of bounded randomized linear contracts is strictly smaller than the revenue achieved by the optimal menu of unbounded randomized linear contracts. 
\end{lemma}
We provide the instance below but defer the full proof to the appendix. 
\begin{proof}[Proof Sketch]
Since we are only concerned with linear contracts, it suffices to specify the utility functions $U_1(\alpha), U_2(\alpha): [0,\infty) \to [0,\infty)$.  In fact it suffices to specify the derivatives $U_i'(\alpha)$ (since we always have that $U_t(\alpha) = \int_{0}^{\alpha}U'_t(x) dx$). 

We specify the slopes below:
\begin{align*}
&U_1'(\alpha) = 
\begin{cases}
\frac{1}{5}  &\text{ if } \alpha \in [0,\frac23) \\
\frac{1}{4}  &\text{ if } \alpha \in [\frac23,\frac43) \\
1 &\text{ if } \alpha \in [\frac43,\infty) 
\end{cases}
\qquad 
&U_2'(\alpha) = 
\begin{cases}
\frac{1}{20} &\text{ if } \alpha \in [0,\frac13) \\
\frac{1}{4} &\text{ if } \alpha \in [\frac13,\frac23) \\
\frac{7}{20} &\text{ if } \alpha \in [\frac23,\infty) 
\end{cases}
\end{align*}

We show (computationally) that in the above instance, there exists a menu of unbounded randomized linear contracts which extracts a revenue strictly greater than $31/100$. On the other hand, the optimal menu of bounded randomized linear contracts extracts a revenue less than $31/100$.
\end{proof}

\subsection{Computing the optimal menu of randomized linear contracts}\label{sec:rlc_alg}

We now switch our focus to the computational problem of efficiently computing the optimal menu of randomized linear contracts. Ultimately, we will show that we can find this optimal menu by solving a concise (polynomially-sized) linear program.

We begin by writing out a much larger program in terms of the variables $\gamma^{(i)}$ (where recall that each $\gamma^{(i)}:[0, \infty) \rightarrow [0, \infty)$ encodes a distribution over the positive reals via its pdf). The optimal menu will satisfy the following mathematical program:
\begin{align}
\sup_{\gamma^{(1)}, \dots, \gamma^{(T)}}
\sum_{i=1}^T &\int_{b=0}^{\infty} U'_i(b)(1-b) \gamma^{(i)}(b) db  \tag{Rev} \label{eqn:rev} \\
&\int_{b=0}^{\infty} \gamma^{(i)}(b) db = 1 \qquad  \forall i \tag{Prob} \label{eqn:prob}\\
&\int_{b=0}^{\infty} U_i(b) ( \gamma^{(i)}(b) - \gamma^{(i')}(b) ) db \geq 0 \qquad \forall i,i' \tag{IC} \label{eqn:ic} \\
& \gamma^{(i)}(b) \geq 0 \qquad \forall b \in \R_{\geq 0} \notag
\end{align}

The first constraint simply states that each $\gamma^{(i)}$ is a probability distribution. \Cref{eqn:ic} is the incentive compatibility constraint which enforces type that $t$ will choose the contract $\gamma^{(t)}$. 
Finally, the objective (Equation~\ref{eqn:rev}) maximizes the expected revenue for the principal over all possible menus of randomized linear contracts. 

In the following theorem, we show how to rewrite this program  as a concise linear program (relying heavily on our characterization in Lemma \ref{lem:brkpoint} of the support of these distributions). 


\begin{theorem} \label{thm:polytime_menu_rlc}
Given an instance principal agent problem with $T$ types and $n$ actions per type with break points $\alpha_0=0\leq \dots \alpha_p $, we can compute a menu $\Gamma = (\gamma^{(1)},\dots,\gamma^{(T)})$ which achieves the optimal revenue in time that is polynomial in $\poly(T,n,|F|)$ where $|F|$ denotes the bit-complexity of the probability matrix. 
\end{theorem}
\begin{proof}
We can write a linear program which computes the optimal menu by essentially solving the mathematical program described above. Using~\Cref{lem:brkpoint}, we know that each randomized linear contract will only put mass on the breakpoints and one additional point $\alpha^t_{p+1}$ above the last break point. One difficulty is that we don't know the value (or even have a bound on) the final break point $\alpha^i_{p+1}.$  To get around this, we instead track the variables $z_t:= \alpha^i_{p+1}\cdot \gamma^{(i)}_{p+1}$ and $\gamma^{(i)}_{p+1}$ (i.e., the total mass above $\alpha_p$ and the average coefficient). 

We can therefore represent each menu as a vector $(\gamma^{(t)}_0,\dots,\gamma^{(t)}_{p},\gamma^{(t)}_{p+1} )$ where $\gamma^{(t)}_i$ denotes the probability of selecting a contract with transfer coefficient $\alpha_i$ when $i\leq p$. $\gamma^{(t)}_{p+1}$ denotes the probability of choosing a point higher than the breakpoints.
Rewriting the original mathematical program using the above simplification, we get the following linear program:

\begin{align}
  & \max & \sum_{t} \sum_{i=0}^p \gamma^{(t)}_i \cdot U'_t(\alpha_i) (1-\alpha_i) +  U'_t(\alpha_p) \left( \gamma^{(t)}_{p+1}  - z^t\right)  \tag{Primal} \label{eqn:lp-primal} \\
  & s.t. & \sum_{i=0}^p U^t(\alpha_i) \cdot (\gamma^{(t)}_i - \gamma^{(t')}_i ) + U'_t(\alpha_p)  \left( z_t - z_{t'}  - \alpha_p \cdot \gamma^{(t)}_{p+1} +\alpha_p \cdot \gamma^{(t')}_{p+1} \right) \notag \\
  & & {} + U_t(\alpha_p) (\gamma^{(t)}_{p+1} - \gamma^{(t)}_{p+1} ) &\geq 0 \notag \\
  & & \alpha_p \cdot \gamma^{(t)}_{p+1} &\leq z^t \notag \\
  & & \sum_{i=0}^{p+1} \gamma^{(t)}_{i} &= 1 \notag
\end{align}

Furthermore, any solution the above linear program can be converted to an optimal menu $\Gamma = (\gamma^{(1)},\dots,\gamma^{(T)})$ with the same objective value. Simply let $\alpha^t_{p+1}= z_t/\gamma^{(t)}_{p+1}$ and $\gamma^{(t)}_{p+1}$ denotes the probability that the contract for type $t$ plays the transfer coefficient $\alpha^t_{p+1}.$ The resulting menu satisfies the original mathematical program and has the same revenue as the objective. 
\end{proof}

One corollary of this proof is that it implies an upper bound on the maximum transfer coefficient (although not a particularly strong one). 

\begin{corollary}
There is an upper bound $B$ which is polynomial in the size of the $\PrincipalAgentProblem$ such that the optimal menu will use a transfer coefficient whose size is at most $B$. 
\end{corollary}
\begin{proof}
Since $z_t$ and $\gamma^{(t)}_{p+1}$ are both variables in the linear program, we know that their size is bounded by $\poly(n,T,|F|).$ Finally, the breakpoint $\alpha^t_{p+1}$ is a ratio of these quantities and therefore is also bounded by $\poly(n,T,|F|).$
\end{proof}

Although we know that the largest transfer is bounded in bit complexity and can be computed in polynomial time, it remains an open question how powerful these  menus of bounded randomized linear contracts really are.
\begin{displayquote}
\textbf{Open Question}: What is the gap in revenue between a menu of bounded randomized linear contracts when compared to menus of unbounded randomized linear contracts? What is the best upper bound on the value of the largest transfer coefficient in the support of any contract?
\end{displayquote}

%
%
%

\label{subsec:linear-vs-rand-menu-linear}

\subsection{Upper bounds for menus of randomized linear contracts}\label{sec:rlc_ub}

While Theorem \ref{thm:rlc_gap} demonstrates a super-constant gap between linear contracts and menus of randomized linear contracts, the instance in the proof require both the number of actions and the number of types to grow without bound. It is natural to ask how this gap depends on each of these parameters -- the number of actions and the number of types -- individually. For example:
\begin{displayquote}
\textbf{Open Question:} What is the dependence on the gap between randomized linear contracts and single linear contracts as a function of the number of actions $n$ (keeping the number of types fixed)? As a function of the number of types $T$ (keeping the number of actions fixed)?
\end{displayquote}

We do not answer this question in generality -- however, interestingly, we show that in the special case where the number of actions is fixed to $2$ (one of the simplest non-trivial cases of the above question), the gap is bounded by a (universal) constant. 
\begin{theorem} \label{thm:Omega_1_lower}
  Suppose we have a principal-agent problem $\PrincipalAgentProblem{}$ with $n = 2$ actions (a null action and a non-null action). Then linear contracts can extract a constant fraction of the profit that randomized menus of linear contracts can:
  \begin{align*}
    \OLinear \PrincipalAgentProblem \ge \Omega(1)\cdot \ORndMenuLinear \PrincipalAgentProblem.
  \end{align*}
\end{theorem}

\begin{proof}[Proof Sketch]
  We bucket the types by the breakpoint of $\alpha$ that shifts the agent to taking its non-null action. If a linear contract does not do well, then these breakpoints must be somewhat spread out. We then argue about the performance of a menu of randomized linear contracts; if it extracts a lot of profit from a type with high breakpoint then that menu item gives a lot of utility to a type with low breakpoint. If it repeatedly extracts a lot of profit over varied breakpoints, then eventually it must be playing breakpoints with greater than one probability, which is a contradiction.
  
  We defer the full proof of Theorem \ref{thm:Omega_1_lower} to Appendix \ref{sec:Omega_1_lower}.
\end{proof}
%
\subsection{Lower bound for menus of randomized linear contracts}\label{sec:rlc_gap}

We will now establish a lower bound on the gap in power between single linear contracts and menus of randomized linear contracts. Specifically, we will exhibit a family of instances with $T$ types and $O(T)$ actions per type where this gap is at least $\Omega(T)$. Note that since this gap is bounded above by $O(T)$ (by simply choosing the best linear contract for one of the $T$ agents; see Lemma \ref{thm:Omega_T_lower}), this lower bound is asymptotically tight in $T$. We provide this construction in the following theorem.

\begin{theorem}\label{thm:rlc_gap}
    There is a principal-agent problem $\PrincipalAgentProblem$ with $T$ agent types and $2T+3$ actions per agent, for which the optimal menu of randomized linear contracts can extract $\Omega(T)$ times as much profit as the optimal single linear contract: 
    \[\ORndMenuLinear \PrincipalAgentProblem\ge \Omega(T) \cdot \OLinear \PrincipalAgentProblem.\]
\end{theorem}
\begin{proof}
Since we are only concerned with benchmarks that pertain to linear menus, to specify a typed principal-agent problem it suffices to specify the utility function $\Utility_t(\alpha):[0, 1] \rightarrow \mathbb{R}_{\geq 0}$ for each agent $t \in [T]$ and the distribution over types (which we will take to be uniform). Since we want each agent to have at most $n = 2T + 3$ actions, each function $\Utility_t(\alpha)$ must be a convex, increasing piecewise linear function with at most $n$ pieces. In our construction, all $T$ such functions will share the same breakpoints, which we label $\alpha_1, \alpha_2, \dots, \alpha_{n-1}$ in increasing order. We will set $\alpha_i = 1 - \eta^i$, where $\eta = 1/2T$. For notational convenience, we'll let $\alpha_0 = 0$ and $\alpha_{n} = 1$.

For simplicity, it will be easier to work with the derivatives $\Utility'_t(\alpha)$ of $\Utility_t(\alpha)$; note that since $\Utility_t(0) = 0$ for each type $t$, each $\Utility_t$ is completely specified by $\Utility'_t$ via the relation $\Utility_t(\alpha) = \int_{0}^{\alpha}\Utility'_{t}(x)dx$. Each $\Utility'_t(\alpha)$ is an increasing piecewise constant function with the same breakpoints $\alpha_i$ as the $\Utility_t$ functions. We will construct them from the following pieces: for each $0 \leq i < n$, define $u_i(\alpha)$ to be the piecewise constant function defined to equal

$$u_i(\alpha) = \frac{1}{10n} \cdot (\alpha_{i+1} - \alpha_i) = \frac{\eta^{-i}}{10n(1-\eta)},$$

\noindent
for $\alpha \in [\alpha_i, \alpha_{i+1})$ and $0$ otherwise. (Note: the second equality above is slightly inaccurate in the case where $i = n-1$, where the RHS should instead equal $\eta^{-(n-1)}/10n$; this will not affect any of the calculations that follow).

We now construct the functions $\Utility'_t$ as follows. Each $\Utility'_t$ will be a slight ``perturbation'' of the sum of all the functions $u_i$; in particular, we have that

\begin{equation}\label{eq:u_der_spec}
    \Utility'_t(\alpha) = \left(\sum_{i=0}^{n-1}u_i(\alpha)\right) + \left(T\cdot u_{t-1}(\alpha) + \frac{T}{2}\cdot u_{t}(\alpha)\right) + \left(T\cdot u_{n-t-2}(\alpha) + \frac{T}{2}\cdot u_{n-t-1}(\alpha)\right).
\end{equation}

\noindent
Note that each such function is indeed increasing (since $\eta = 1/2T$, $T\eta^{-i} < \eta^{-(i+1)}$) so these specify valid utility functions. One useful way of thinking about this example is that we divide our actions into two blocks of approximately $T$ actions each. In $\Utility_t(\alpha)$, each action corresponds to a linear piece with a fixed height of $1/10n$ (since $\int_{\alpha_i}^{\alpha_{i+1}}u_i(x)dx = 1/(10n)$ is independent of $i$), except for four of these actions, which have heights that are $\Theta(T)$ times larger. Moreover, these four actions are chosen symmetrically; two consecutive actions from the first block of $T$ actions, and their ``mirror image'' from the second block of $T$ actions. 

To prove a lower bound on the gap provided by this instance, we will begin by proposing a menu of randomized linear contracts, and show that each type is incentivized to choose a specific contract from this menu. Recall that $\gamma^{(t)}$ denotes the randomized linear contract that type $t$ will choose from the menu. We will let $\gamma^{(t)}$ be the uniform distribution of support size $2$ that places $50\%$ of its weight on the point $\alpha_{t}$ and the other $50\%$ of its weight on the point $\alpha_{n-t-1}$.

We now compute the utility achieved by this type (and other types) who receive this contract. To begin, note that

$$\int_{0}^{\alpha_{i}}\left(\sum_{i=0}^{n-1}u_{i}(x)\right)dx = \frac{i}{10n}.$$

In particular, this means that

$$\E_{\alpha \sim \gamma^{(t)}}\left[\int_{0}^{\alpha}\left(\sum_{i=0}^{n-1}u_{i}(x)\right)dx\right] = \frac{1}{20n},$$

\noindent
or in other words, the total contribution to each agent's utility from the $\sum u_i(\alpha)$ term in \eqref{eq:u_der_spec} is independent of the contract $\gamma^{(t)}$ they select. It thus suffices to only consider the latter four terms of \eqref{eq:u_der_spec}.

Now, if the type $t$ agent selects contract $\gamma^{(t)}$, they receive an extra $1.75T/10n$ utility in expectation: when they draw the $\alpha_{t}$ contract, they receive $T/10n$ extra utility, and when they draw the $\alpha_{n-t-1}$ contract, they receive $(T + T/2 + T)/10n$ extra utility. On the other hand, if the type $t$ agent selects contract $\gamma^{(t')}$ for any $t' \neq t$, we claim they only receive an extra $1.5T/10n$ utility in expectation. Specifically, if $t' > t$, then they are guaranteed to receive $(1.5T/10n)$ utility from the first two terms of \eqref{eq:u_der_spec} but never receive any utility from the second two terms of \eqref{eq:u_der_spec} (regardless of which contract they draw). Similarly, if $t' < t$, then they receive no extra utility when they draw the $\alpha_{t'}$ contract, but receive $3T/10n$ utility (from all four terms of \eqref{eq:u_der_spec}) when they draw the $\alpha_{n-t'-1}$ contract. This shows it is strictly dominant for an agent of type $t$ to select the contract $\gamma^{(t)}$.

We now wish to examine the profit received by both i. this menu of randomized linear contracts and ii. the optimal single linear contract and show that the gap between these two values is indeed $\Omega(T)$. We begin with the profit obtained by this menu. Recall that the profit $\TypedProfit{t}(\alpha)$ obtained by offering a linear contract $\alpha$ to an agent of type $t$ is equal to $(1-\alpha)\Utility'_t(\alpha)$. Now, note that (for any $0 \leq i < n-1$) that 

\begin{equation}\label{eq:profit_bound}
(1-\alpha_i)u_i(\alpha_i) = \eta^{i}u_i(\alpha_i) = \frac{1}{10n(1-\eta)}.
\end{equation}

So, the expected profit obtained by offering contract $\gamma^{(t)}$ to type $t$ is given by

\begin{eqnarray*}
\TypedProfit{t}(\gamma^{(t)}) &=& \frac{1}{2}(1-\alpha_{t})\Utility'_t(\alpha_t) + \frac{1}{2}(1-\alpha_{n-t-1})\Utility'_t(\alpha_{n-t-1})\\
&=& \left(1 + \frac{T}{2}\right)\cdot \frac{1}{10n(1-\eta)} \\
&=& \Omega(1).\end{eqnarray*}

\noindent
The total profit from this menu is therefore $\frac{1}{T}\sum_{t=1}^{T}\TypedProfit{t}(\gamma^{(t)}) = \Omega(1)$. On the other hand, if we offer a single linear contract $\alpha$ to all $t$ types, we receive profit

$$\Profit(\alpha) = (1-\alpha)\frac{1}{T}\sum_{t=1}^{T}\Utility'_{t}(\alpha).$$

\noindent
But now, by \eqref{eq:u_der_spec}, note that

$$\frac{1}{T}\sum_{t=1}^{T} U'_t(\alpha) \leq \frac{5}{2} \left(\sum_{i=0}^{n-1}u_{i}(\alpha)\right)$$

\sloppy{But now, $(1-\alpha)\left(\sum_{i=0}^{n-1}u_{i}(\alpha)\right)$ is maximized at one of the $\alpha_i$, where by \eqref{eq:profit_bound} it is at most $\frac{1}{10n(1-\eta)}$. It follows that $\Profit(\alpha) \leq (5/2)/(10n(1-\eta)) = O(1/T)$, and therefore that $\ORndMenuLinear \PrincipalAgentProblem\ge \Omega(T) \cdot \OSingle \PrincipalAgentProblem$.}

\end{proof}

As a corollary of Theorem \ref{thm:rlc_gap}, we can show the same gap between menus of deterministic general contracts ($\ODetMenu$) and menus of randomized general contracts ($\ORndMenu$).

\begin{corollary}\label{cor:det_vs_rnd}
There is a principal-agent problem $\PrincipalAgentProblem$ with $T$ agent types and $2T+3$ actions per agent, for which the optimal menu of randomized (general) contracts can extract $\Omega(T)$ times as much profit as the optimal single (general) contract: 
\[\ORndMenu \PrincipalAgentProblem\ge \Omega(T) \cdot \ODetMenu \PrincipalAgentProblem.\]
\end{corollary}

The proof of Corollary \ref{cor:det_vs_rnd} is based on the observation that in principal-agent problems with $m=2$ outcomes, all general contracts are very close to being linear contracts (in fact, they are affine contracts). We defer the proof to Appendix \ref{apx:det_vs_rnd}. 

\bibliography{contracts}
\appendix
\section{Omitted proofs}

\subsection{Proof of Theorem \ref{thm:Omega_n_gap}}
\label{subsec:apx-omega-n-gap}

Here we provide the detailed proof of Theorem \ref{thm:Omega_n_gap} (that there exists an $\Omega(n)$ gap between single contracts and menus of deterministic contracts). We begin by establishing some lemmas.

\begin{lemma}\label{lemma:x4_upper_bound_Omega_n_gap}
Consider the principal-agent problem $\PrincipalAgentProblem$ defined in the above $\Omega(n)$-gap instance. For any contract $\ContractVector=(x_1,x_2,x_3,x_4)\in\mathbb{R}_{\ge 0}^4$ that has non-negative profit, it must hold that $x_4\le\frac{5}{n^{3n}}$.
\end{lemma}
\begin{proof}
First, note that for type $1$ agent, only first $n$ actions can generate strict positive welfare, and the welfare of action $i\in[n]$ is $F^{(1)}_{i,2}-c_i=\frac{i}{n^{n+1}}$, and hence the maximum welfare of type $1$ agent is $\frac{1}{n^n}$ generated by action $n$. For type $2$ agent, only actions $\{n+1,\dots,2n\}$ can generated strict positive welfare, and the welfare of action $n+i$ for $i\in[n]$ is $F^{(2)}_{n+i,3}-c_{n+i}=\frac{4i}{n^{n+1}}$, and hence the maximum welfare of type $2$ agent is $\frac{4}{n^n}$ which is generated from action $2n$. Therefore, the maximum possible profits a contract can get from type $1$ agent and type $2$ agent are $\frac{1}{n^n}$ and $\frac{4}{n^n}$ respectively.

Notice that regardless of which action type $3$ agent plays, the result is always outcome $4$ which has zero reward, and thus, the contract $\ContractVector=(x_1,x_2,x_3,x_4)$ pays type $3$ agent $x_4$ for nothing in return.

Therefore, the profit in expectation over agent types of the contract $\ContractVector$ is at most
$\Pr[\textrm{agent type}=1]\cdot\frac{1}{n^n}+\Pr[\textrm{agent type}=2]\cdot\frac{4}{n^n}-\Pr[\textrm{agent type}=3]\cdot x_4$, and for this to be non-negative, it must hold that
\[
    x_4\le\frac{\Pr[\textrm{agent type}=1]\cdot\frac{1}{n^n}+\Pr[\textrm{agent type}=2]\cdot\frac{4}{n^n}}{\Pr[\textrm{agent type}=3]}=\frac{\frac{1}{2n^{2n}}\cdot\frac{1}{n^n}+\frac{1}{2n^{2n}}\cdot\frac{4}{n^n}}{1-\frac{1}{n^{2n}}}\le\frac{5}{n^{3n}}.
\]
\end{proof}

Because of Lemma~\ref{lemma:x4_upper_bound_Omega_n_gap}, \emph{we henceforth only consider} contract $\ContractVector$ with $x_4\le\frac{5}{n^{3n}}$ for the $\Omega(n)$-gap instance. As a consequence, if an action only results in outcome $4$ for a type of agent, then the resulting payment to that agent is at most $x_4\le\frac{5}{n^{3n}}$, which is strictly less than the cost of any action $i\in[2n]$. In the construction of $\Omega(n)$-gap instance, we mentioned that outcome $4$ always occurs when type $1$ agent plays the last $n+1$ actions and when type $2$ agent plays the first $n$ actions and when type $3$ agent plays any action. Thus, \emph{we can henceforth assume} that type $1$ agent never plays actions $\{n+1,\dots,2n\}$, and type $2$ agent never plays actions $[n]$, and type $3$ agent never plays actions $[2n]$.

\begin{lemma}\label{lemma:agent_2_profit_upper_bound_Omega_n_gap}
Consider the principal-agent problem $\PrincipalAgentProblem$ defined in the above $\Omega(n)$-gap instance. For any contract $\ContractVector$ with non-negative profit, we have
\begin{equation}\label{eq:agent_2_profit_upper_bound_Omega_n_gap}
    \TypedProfit{2}(\CostVector, \ForecastTensor, \RewardVector, \ContractVector)\le \frac{9}{n^{n+1}}-x_2\cdot F^{(2)}_{2n+1,2}.
\end{equation}
\end{lemma}
\begin{proof}
Consider any contract $\ContractVector=(x_1,x_2,x_3,x_4)\in\mathbb{R}_{\ge 0}^4$. First, since $F^{(2)}_{n+i,1}=0$ for all $i\in[n+1]$, type $2$ agent's utility and the principal's profit from type $2$ agent are both independent of $x_1$, and thus, we can ignore outcome 1 and assume $x_1=0$ without loss of generality.

Moreover, since $F^{(2)}_{n+i,2}$ is the same for all $i\in[n+1]$, the expected payment type $2$ agent gets from outcome 2 is $x_2$ times a constant regardless of which non-null action type $2$ agent plays. Hence, among all non-null actions (also note that action $2n+1$ is always no worse than the null action for type $2$ agent), type $2$ agent prefers the non-null action $n+\ell$ that maximizes the expected payment for outcome 3 minus the cost of the action, i.e., $\ell=\argmax_{i\in[n+1]} x_3\cdot F^{(2)}_{n+i,3}+x_4\cdot F^{(2)}_{n+i,4}-c_i$.

Note that $\ell$ does not depend on $x_2$ at all. Thus, type $2$ agent will still prefers the same action when faced with contract $(0,0,x_3,x_4)$ instead of $(0,x_2,x_3,x_4)$. Therefore, it remains to prove that the profit which contract $(0,0,x_3,x_4)$ gets from type $2$ agent is at most $\frac{9}{n^{n+1}}$ (which implies that the profit of contract $(0,x_2,x_3,x_4)$ is at most $\frac{9}{n^{n+1}}-x_2\cdot F^{(2)}_{2n+1,2}$).

Henceforth, we further assume that
\begin{equation}\label{eq:i*_definition_Omega_n_gap}
    \ell=\argmax_{i\in[n]} x_3\cdot F^{(2)}_{n+i,3}+x_4\cdot F^{(2)}_{n+i,4}-c_i,
\end{equation}
i.e., we ignore action $2n+1$ without loss of generality, because if type $2$ agent prefers action $2n+1$, then it generates zero reward and hence zero profit for the principal. Next, we show that for any $x_3,x_4$, the principal's profit from type $2$ agent is at most $\frac{9}{n^{n+1}}$ when type $2$ agent plays $n+\ell$.

We first prove a necessary condition for type $2$ agent to play action $n+\ell$ for $\ell\in\{2,3,\dots,n\}$: If type $2$ agent prefers playing action $n+\ell$, then it must hold that $x_3\ge1-\frac{1}{n^{\ell}-n^{\ell-1}}-\frac{5}{4n^{2n}}$.

To this end, note that if type $2$ agent prefers playing action $n+\ell$, then in particular, type $2$ agent prefers action $n+\ell$ over $n+\ell-1$, which by Eq.~\eqref{eq:i*_definition_Omega_n_gap} implies that $$x_3\cdot F^{(2)}_{n+\ell,3}+x_4\cdot F^{(2)}_{n+\ell,4}-c_{\ell}\ge x_3\cdot F^{(2)}_{n+\ell-1,3}+x_4\cdot F^{(2)}_{n+\ell-1,4}-c_{\ell-1}.$$ After rearrangement, this is equivalent to 
\begin{align*}
    x_3&\ge \frac{c_{\ell}-c_{\ell-1}+x_4\cdot(F^{(2)}_{n+\ell-1,4}-F^{(2)}_{n+\ell,4})}{F^{(2)}_{n+\ell,3}-F^{(2)}_{n+\ell-1,3}}\\
    &\ge\frac{c_{\ell}-c_{\ell-1}-x_4}{F^{(2)}_{n+\ell,3}-F^{(2)}_{n+\ell-1,3}}&&\text{($F^{(2)}_{n+\ell-1,4},F^{(2)}_{n+\ell,4}\in[0,1]$)}\\
    &\ge\frac{c_{\ell}-c_{\ell-1}-\frac{5}{n^{3n}}}{F^{(2)}_{n+\ell,3}-F^{(2)}_{n+\ell-1,3}}&&\text{(By Lemma~\ref{lemma:x4_upper_bound_Omega_n_gap})}\\
    &=\frac{\frac{4(n^{\ell}-n^{\ell-1}+1)}{n^{n+1}}-\frac{5}{n^{3n}}}{\frac{4(n^{\ell}-n^{\ell-1})}{n^{n+1}}}&&\text{(By definition of $F^{(2)}_{n+\ell,3},F^{(2)}_{n+\ell-1,3},c_{\ell},c_{\ell-1}$)}\\
    &\ge 1-\frac{1}{n^{\ell}-n^{\ell-1}}-\frac{5}{4n^{2n}}&&\text{($n^{\ell}-n^{\ell-1}\ge n$ for any $\ell\in\{2,3,\dots,n\}$)}.
\end{align*}

If $\ell=1$, type $2$ agent plays action $n+1$, which generates welfare $F^{(2)}_{n+1,3}\cdot 1-c_{n+1}=\frac{4}{n^{n+1}}$, and obviously, the principal's profit from type $2$ agent cannot exceed this. Hence, we consider the case when type $2$ agent plays action $n+\ell$ for $\ell\in\{2,3,\dots,n\}$. In this case, as we have shown, the necessary condition for type $2$ agent to prefer action $n+\ell$ is $x_3\ge1-\frac{1}{n^{\ell}-n^{\ell-1}}-\frac{5}{4n^{2n}}$. The principal's profit from type $2$ agent is at most $(1-x_3)\cdot F^{(2)}_{n+\ell,3}=(1-x_3)\cdot\frac{4}{n^{n-\ell+1}}$, which is at most $(\frac{1}{n^{\ell}-n^{\ell-1}}+\frac{5}{4n^{2n}})\cdot \frac{4}{n^{n-\ell+1}}=\frac{4}{n^{n+1}-n^{n}}+\frac{5}{n^{3n-\ell+1}}\le\frac{8}{n^{n+1}}+\frac{1}{n^{n+1}}=\frac{9}{n^{n+1}}$.
\end{proof}

\begin{lemma}\label{lemma:agent_1_low_profit_for_large_x1_Omega_n_gap}
Consider the principal-agent problem $\PrincipalAgentProblem$ defined in the above $\Omega(n)$-gap instance. For any contract $\ContractVector=(x_1,x_2,x_3,x_4)\in\mathbb{R}_{\ge0}^4$ with non-negative profit, we have
\begin{equation}\label{eq:agent_1_low_profit_for_large_x1_Omega_n_gap}
    \TypedProfit{1}(\CostVector, \ForecastTensor, \RewardVector, \ContractVector)\le \frac{3}{n^{n+1}}+\frac{2(1-x_1)}{n^n}.
\end{equation}
\end{lemma}
\begin{proof}
Consider any contract $\ContractVector=(x_1,x_2,x_3,x_4)\in\mathbb{R}_{\ge0}^4$. Since $F^{(1)}_{i,3}=0$ for all $i\in[n]$, type $1$ agent's utility and the principal's profit from type $1$ agent are both independent of $x_3$, and thus, we can ignore outcome 3 and assume $x_3=0$ without loss of generality. Also, notice that outcome $2$ and $4$ have zero reward and hence do not contribute to the principal's profit.

Given contract $\ContractVector$, let $\ell$ be the action that type $1$ agent prefers. Without loss of generality, we assume that $\ell$ is not the null action or action $2n+1$ (because otherwise Eq.~\eqref{eq:agent_1_low_profit_for_large_x1_Omega_n_gap} holds trivially). Notice that the welfare generated by action $\ell$ is $F^{(1)}_{\ell,1}\cdot 1-c_{\ell}$, which is 
equal to $\frac{\ell}{n^{n+1}}$ by definition of $F^{(1)}_{\ell,1},c_{\ell}$, and we define $\delta\in[0,1]$ such that $\TypedProfit{1}(\CostVector, \ForecastTensor, \RewardVector, \ContractVector)=\frac{\delta\ell}{n^{n+1}}$. Hence, our goal is to upper bound $\frac{\delta\ell}{n^{n+1}}$. Without loss of generality, we also assume that $\ell\ge 2$ (because the welfare generated by action $1$ is $F^{(1)}_{1,1}\cdot 1-c_{1}=\frac{1}{n^{n+1}}$, and the principal's profit from type $1$ agent cannot exceed this).

Since type $1$ agent plays action $\ell$ given contract $\ContractVector$, we have that $\TypedProfit{1}(\CostVector, \ForecastTensor, \RewardVector, \ContractVector)=(1-x_1)\cdot F^{(1)}_{\ell,1}-x_2\cdot F^{(1)}_{\ell,2}-x_4\cdot F^{(1)}_{\ell,4}$. Since $\TypedProfit{1}(\CostVector, \ForecastTensor, \RewardVector, \ContractVector)=\frac{\delta\ell}{n^{n+1}}$ by definition of $\delta$, it follows that
\begin{equation}
x_2\cdot F^{(1)}_{\ell,2}=(1-x_1)\cdot F^{(1)}_{\ell,1}-x_4\cdot F^{(1)}_{\ell,4}-\frac{\delta\ell}{n^{n+1}}.\label{eq:x_2_for_delta_Omega_n_gap}
\end{equation}

Now we derive the utility generated by any action $i\in[n]$ for type 1 agent (denoted by $u_{i}^{(1)}$):
\begin{align}
u_{i}^{(1)}=&x_1\cdot F^{(1)}_{i,1}+x_2\cdot F^{(1)}_{i,2}+x_4\cdot F^{(1)}_{i,4}-c_i\nonumber\\
=&F^{(1)}_{i,1}-c_i - (1-x_1)\cdot F^{(1)}_{i,1}+x_2\cdot F^{(1)}_{i,2}+x_4\cdot F^{(1)}_{i,4}\nonumber\\
=&\frac{i}{n^{n+1}}-\frac{1-x_1}{n^{n+1-i}}+x_2\cdot F^{(1)}_{i,2}+x_4\cdot F^{(1)}_{i,4} &&\text{(By definition of $F^{(1)}_{i,1},c_i$)}\nonumber\\
=&\frac{i}{n^{n+1}}-\frac{1-x_1}{n^{n+1-i}}+x_2\cdot F^{(1)}_{\ell,2}\cdot\frac{n^i-i}{n^{\ell}-\ell}+x_4\cdot F^{(1)}_{i,4} &&\text{(By definition of $F^{(1)}_{i,2},F^{(1)}_{\ell,2}$)}\nonumber\\
=&\frac{i}{n^{n+1}}-\frac{1-x_1}{n^{n+1-i}}+\frac{n^i-i}{n^{\ell}-\ell}\cdot\left((1-x_1)\cdot F^{(1)}_{\ell,1}-\frac{\delta\ell}{n^{n+1}}\right)\nonumber\\
&\qquad\qquad\qquad\qquad+x_4\cdot (F^{(1)}_{i,4}-\frac{n^i-i}{n^{\ell}-\ell}\cdot F^{(1)}_{\ell,4}) &&\text{(By Eq.~\eqref{eq:x_2_for_delta_Omega_n_gap})}.\label{eq:utility_action_i_type_1_Omega_n_gap}
\end{align}

Using Eq.~\eqref{eq:utility_action_i_type_1_Omega_n_gap}, we calculate the the utility generated by action $\ell$ for type 1 agent:
\begin{align*}
u_{\ell}^{(1)}=&\frac{\ell}{n^{n+1}}-\frac{1-x_1}{n^{n+1-\ell}}+\left((1-x_1)\cdot F^{(1)}_{\ell,1}-\frac{\delta\ell}{n^{n+1}}\right) &&\text{(By Eq.~\eqref{eq:utility_action_i_type_1_Omega_n_gap})}\\
=&(1-\delta)\cdot\frac{\ell}{n^{n+1}} &&\text{(By definition of $F^{(1)}_{\ell,1}$)},
\end{align*}
and we lower bound the the utility generated by action $\ell-1$ for type 1 agent:
\begin{align*}
&u_{\ell-1}^{(1)}-x_4\cdot (F^{(1)}_{\ell-1,4}-\frac{n^{\ell-1}-\ell+1}{n^{\ell}-\ell}\cdot F^{(1)}_{\ell,4})\\
=&\frac{\ell-1}{n^{n+1}}-\frac{1-x_1}{n^{n+2-\ell}}+\frac{n^{\ell-1}-\ell+1}{n^{\ell}-\ell}\cdot\left((1-x_1)\cdot F^{(1)}_{\ell,1}-\frac{\delta\ell}{n^{n+1}}\right) &&\text{(By Eq.~\eqref{eq:utility_action_i_type_1_Omega_n_gap})}\\
=&\frac{\ell-1}{n^{n+1}}-\frac{1-x_1}{n^{n+2-\ell}}+\frac{1}{n}\cdot\left(1-\frac{n(\ell-1)-\ell}{n^{\ell}-\ell}\right)\cdot\left((1-x_1)\cdot F^{(1)}_{\ell,1}-\frac{\delta\ell}{n^{n+1}}\right)\\
=&\frac{\ell-1}{n^{n+1}}-\frac{n(\ell-1)-\ell}{n^{\ell}-\ell}\cdot\frac{1-x_1}{n^{n+2-\ell}}-\frac{1}{n}\cdot\left(1-\frac{n(\ell-1)-\ell}{n^{\ell}-\ell}\right)\cdot\frac{\delta\ell}{n^{n+1}} &&\text{(By definition of $F^{(1)}_{\ell,1}$)}\\
\ge&\frac{\ell-1}{n^{n+1}}-\frac{n(\ell-1)-\ell}{n^{\ell}-\ell}\cdot\frac{1-x_1}{n^{n+2-\ell}}-\frac{1}{n}\cdot\frac{\delta\ell}{n^{n+1}} &&\text{($\frac{n(\ell-1)-\ell}{n^{\ell}-\ell}\le 1$ as $\ell\ge 2$)}\\
\ge&\frac{\ell-1}{n^{n+1}}-\frac{n(\ell-1)-\ell}{n^{\ell}-\ell}\cdot\frac{1-x_1}{n^{n+2-\ell}}-\frac{1}{n^{n+1}} &&\text{($\ell\le n$)}\\
=&\frac{\ell-2}{n^{n+1}}-\left(1+\frac{\ell}{n^{\ell}-\ell}\right)\cdot\left(\ell-1-\frac{\ell}{n}\right)\cdot\frac{1-x_1}{n^{n+1}}\\
\ge&\frac{\ell-2}{n^{n+1}}-\left(1+\frac{\ell}{n^{\ell}-\ell}\right)\cdot\frac{1-x_1}{n^{n}} &&\text{($\ell-1-\frac{\ell}{n}\le n$)}\\
\ge&\frac{\ell-2}{n^{n+1}}-2\cdot\frac{1-x_1}{n^n} &&\text{($2\ell\le n^{\ell}$)}.
\end{align*}
Recall that $\ell$ is type $1$ agent's favorite action, and hence, for type $1$ agent,
\[
u_{\ell-1}^{(1)}\le u_{\ell}^{(1)},
\]
which implies that
$
\frac{\ell-2}{n^{n+1}}-2\cdot\frac{1-x_1}{n^n}+x_4\cdot (F^{(1)}_{\ell-1,4}-\frac{n^{\ell-1}-\ell+1}{n^{\ell}-\ell}\cdot F^{(1)}_{\ell,4})\le(1-\delta)\cdot\frac{\ell}{n^{n+1}}$. After rearrangement, this is equivalent to 
\begin{align*}
\frac{\delta\ell}{n^{n+1}}&\le\frac{2}{n^{n+1}}+\frac{2(1-x_1)}{n^n}-x_4\cdot \left(F^{(1)}_{\ell-1,4}-\frac{n^{\ell-1}-\ell+1}{n^{\ell}-\ell}\cdot F^{(1)}_{\ell,4}\right)\\
&\le\frac{2}{n^{n+1}}+\frac{2(1-x_1)}{n^n}+x_4\cdot \frac{n^{\ell-1}-\ell+1}{n^{\ell}-\ell} &&\text{($F^{(1)}_{\ell-1,4},F^{(1)}_{\ell,4}\in[0,1]$)}\\
&\le\frac{2}{n^{n+1}}+\frac{2(1-x_1)}{n^n}+x_4\\
&\le\frac{3}{n^{n+1}}+\frac{2(1-x_1)}{n^n}&&\text{(By Lemma~\ref{lemma:x4_upper_bound_Omega_n_gap})}.
\end{align*}
\end{proof}

\begin{lemma}\label{lemma:lower_bound_of_x2_Omega_n_gap}
Consider the principal-agent problem $\PrincipalAgentProblem$ defined in the above $\Omega(n)$-gap instance. For any contract $\ContractVector=(x_1,x_2,x_3,x_4)\in\mathbb{R}_{\ge0}^4$ with $x_1=1-n^{-\alpha}$ for any $\alpha\le1$, if there is an action $\ell\in\{2,3,\dots,n\}$ that generates non-negative utility for type $1$ agent, then it must hold that
\begin{equation}\label{eq:lower_bound_of_x2_Omega_n_gap}
    x_2\ge \frac{n^{n-1}-1}{F^{(1)}_{n,2}}\cdot \left(\frac{1}{2n^{n+\alpha}}-\frac{5}{n^{3n}}\right).
\end{equation}
\end{lemma}
\begin{proof}
Consider any contract $\ContractVector=(x_1,x_2,x_3,x_4)\in\mathbb{R}_{\ge0}^4$ with $x_1=1-n^{-\alpha}$ with $\alpha\le1$. The utility of action $\ell\in\{2,3,\dots,n\}$ for type $1$ agent is $F^{(1)}_{\ell,1}\cdot x_1+F^{(1)}_{\ell,2}\cdot x_2+F^{(1)}_{\ell,4}\cdot x_4-c_{\ell}$. For this to be non-negative, it must hold that
\begin{align*}
    x_2&\ge\frac{c_{\ell}-F^{(1)}_{\ell,1}+F^{(1)}_{\ell,1}\cdot(1-x_1)-F^{(1)}_{\ell,4}\cdot x_4}{F^{(1)}_{\ell,2}}\\
    &=\frac{-\ell+n^{\ell-\alpha}}{F^{(1)}_{\ell,2}}\cdot\frac{1}{n^{n+1}}-\frac{F^{(1)}_{\ell,4}\cdot x_4}{F^{(1)}_{\ell,2}}&&\text{(By definition of $c_{\ell}, F^{(1)}_{\ell,1}, x_1$)}\\
    &=\frac{n^{n-1}-1}{F^{(1)}_{n,2}}\cdot\left(\frac{n^{\ell-\alpha}-\ell}{n^{\ell}-\ell}\cdot\frac{1}{n^{n}}-\frac{F^{(1)}_{\ell,4}\cdot x_4\cdot n}{(n^{\ell}-\ell)}\right)&&\text{(By definition of $F^{(1)}_{\ell,2}$)}\\
    &\ge\frac{n^{n-1}-1}{F^{(1)}_{n,2}}\cdot\left(\frac{n^{\ell-\alpha}-\ell}{n^{\ell}-\ell}\cdot\frac{1}{n^{n}}-x_4\right) &&\text{($F^{(1)}_{\ell,4}\le1$ and $n^{\ell}-\ell\ge n$ for $\ell\ge2$)}\\
    &\ge\frac{n^{n-1}-1}{F^{(1)}_{n,2}}\cdot\left(\frac{n^{\ell-\alpha}-\ell}{n^{\ell}-\ell}\cdot\frac{1}{n^{n}}-\frac{5}{n^{3n}}\right) &&\text{(By Lemma~\ref{lemma:x4_upper_bound_Omega_n_gap})}.
\end{align*}
It remains to prove that $\frac{n^{\ell-\alpha}-\ell}{n^{\ell}-\ell}\ge\frac{n^{-\alpha}}{2}$. To this end, note that $n^{\ell-1}\ge 2\ell$ for $\ell\ge 2$ and $n\ge 12$, and we derive that
\begin{align*}
    \frac{n^{\ell-\alpha}-\ell}{n^{\ell}-\ell}&=\frac{n^{-\alpha}-\ell/n^{\ell}}{1-\ell/n^{\ell}}\\
    &\ge\frac{n^{-\alpha}-n^{-1}/2}{1-\ell/n^{\ell}} &&\text{(By $n^{\ell-1}\ge 2\ell$)}\\
    &\ge n^{-\alpha}-\frac{n^{-1}}{2}\\
    &\ge \frac{n^{-\alpha}}{2} &&\text{(By $\alpha\le 1$)}.
\end{align*}
\end{proof}

We can now proceed to prove the main theorem.

\begin{proof}[Proof of Theorem~\ref{thm:Omega_n_gap}]\,

\textbf{Soundness.} We prove that any single contract $\ContractVector=(x_1,x_2,x_3,x_4)\in\mathbb{R}_{\ge0}^4$ can only get total profit at most $\frac{14}{n^{n+1}}$ from first two types of agents (recall type $3$ agent generates zero reward). By Lemma~\ref{lemma:agent_2_profit_upper_bound_Omega_n_gap}, contract $\ContractVector$ can only get profit at most $\frac{9}{n^{n+1}}$ from type $2$ agent. By Lemma~\ref{lemma:agent_1_low_profit_for_large_x1_Omega_n_gap}, if $x_1\ge1-\frac{1}{n}$, then contract $\ContractVector$ can only get profit at most $\frac{5}{n^{n+1}}$ from type $1$ agent, and thus, the total profit from two types of agents is at most $\frac{14}{n^{n+1}}$.

Therefore, it remains to analyze the case $x_1=1-n^{-\alpha}$ for $\alpha\le 1$. In this case, we can without loss of generality assume that type $1$ agent prefers an action $\ell\in\{2,\dots,n\}$, because otherwise type $1$ agent prefers either the null action (which generates zero profit) or action $1$ (which generates $F^{(1)}_{1,1}\cdot 1-c_1=\frac{1}{n^{n+1}}$ welfare and hence at most $\frac{1}{n^{n+1}}$ profit). Since type $1$ agent prefers action $\ell\in\{2,\dots,n\}$ over the null action, action $\ell$ must have non-negative utility, and it follows by Lemma~\ref{eq:lower_bound_of_x2_Omega_n_gap} that Ineq.~\eqref{eq:lower_bound_of_x2_Omega_n_gap} holds. Combining Ineq.~\eqref{eq:agent_2_profit_upper_bound_Omega_n_gap} in Lemma~\ref{lemma:agent_2_profit_upper_bound_Omega_n_gap} with Ineq.~\eqref{eq:lower_bound_of_x2_Omega_n_gap}, we have that
\begin{align*}
    \TypedProfit{2}(\CostVector, \ForecastTensor, \RewardVector, \ContractVector)&\le\frac{9}{n^{n+1}}-\frac{F^{(2)}_{2n+1,2}\cdot (n^{n-1}-1)}{F^{(1)}_{n,2}}\cdot\left(\frac{1}{2n^{n+\alpha}}-\frac{5}{n^{3n}}\right)\\
    &=\frac{9}{n^{n+1}}-\left(\frac{2}{n^{n+\alpha}}-\frac{20}{n^{3n}}\right) &&\text{(By definition of $F^{(2)}_{2n+1,2}$)}\\
    &\le\frac{10}{n^{n+1}}-\frac{2}{n^{n+\alpha}}.
\end{align*}
On the other hand, by Lemma~\ref{lemma:agent_1_low_profit_for_large_x1_Omega_n_gap}, we have that
\begin{align*}
\TypedProfit{1}(\CostVector, \ForecastTensor, \RewardVector, \ContractVector)&\le \frac{3}{n^{n+1}}+\frac{2(1-x_1)}{n^n}\\
&=\frac{3}{n^{n+1}}+\frac{2}{n^{n+\alpha}}.
\end{align*}
Hence, the total profit of contract $\ContractVector$ from first two types of agents is at most $\frac{13}{n^{n+1}}$ in this case.

\textbf{Completeness.} We show a menu of two contracts that gets total profit $\frac{1}{n^n}$ from first two types of agents and zero profit from type $3$ agent. The two contracts are $\ContractVector_1=(0,\frac{n^{n-1}-1}{n^{n}\cdot F^{(1)}_{n,2}},0,0)$ and $\ContractVector_2=(0,0,1,0)$. First, it is easy to check that given contract $\ContractVector_1$, every action among $[n]$ has zero utility for type $1$ agent (i.e., $F^{(1)}_{i,2}\cdot\frac{n^{n-1}-1}{n^{n}\cdot F^{(1)}_{n,2}}=c_i$ for all $i\in[n]$), and given contract $\ContractVector_2$, none of the actions can make strictly positive utility for type $1$ agent. Thus, we can assume that type $1$ agent chooses contract $\ContractVector_1$ and plays action $n$, and then the principal's profit from type $1$ agent is $F^{(1)}_{n,1}\cdot 1-F^{(1)}_{n,2}\cdot\frac{n^{n-1}-1}{n^{n}\cdot F^{(1)}_{n,2}}=\frac{1}{n^n}$.

On the other hand, type $2$ agent will get utility $F^{(2)}_{2n,3}\cdot 1-c_{2n}=\frac{4}{n^n}$ by choosing contract $\ContractVector_2$ and playing action $2n$. Moreover, if choosing contract $\ContractVector_1$, type $2$ agent will get utility $F^{(2)}_{\ell,2}\cdot\frac{n^{n-1}-1}{n^{n}\cdot F^{(1)}_{n,2}}-c_{\ell}=\frac{4}{n^n}-c_{\ell}$ for action $\ell\in\{n+1,n+2,\dots,2n+1\}$, which is at most $\frac{4}{n^n}$. Thus, we can assume that type $2$ agent chooses contract $\ContractVector_2$, and then the principal's profit from type $2$ agent is non-negative, because for any outcome, contract $\ContractVector_2$ does not pay more than its reward. It follows that the total profit from two types of agents is at least $\frac{1}{n^n}$.

Finally, recall that regardless of which action type $3$ agent plays, the result is always outcome $4$. Since the payment for outcome $4$ is zero for both contracts $\ContractVector_1,\ContractVector_2$, the profit from type $3$ agent is zero.
\end{proof}

\subsection{Proof of Theorem \ref{thm:Omega_log_T_gap}}

Here we provide the detailed proof of Theorem \ref{thm:Omega_log_T_gap} (that there exists an $\Omega(\log T)$ gap between single contracts and menus of deterministic contracts). 

\begin{lemma}\label{lemma:x_T_upper_bound_Omega_log_T_gap}
Consider the principal-agent problem $\PrincipalAgentProblem$ defined in the above $\Omega(\log T)$-gap instance. For any contract $\ContractVector=(x_1,\dots,x_{T})\in\mathbb{R}_{\ge 0}^T$ that has non-negative profit, it must hold that $x_T\le\frac{1}{16^N}$.
\end{lemma}
\begin{proof}
Note that only outcome 1 has strictly positive reward, and only action $1$ can result in outcome 1. Hence, the welfare of type $((k-1)2^N+\ell)$ agent is $F^{((k-1)2^N+\ell)}_{1,1}\cdot r_1 - c_1=\frac{2^{-k}}{2T(1-2^{-k})}$ for all $k\in[N]$ and $\ell\in[2^N]$, and the principal's profit from this type of agent is at most this amount.

Moreover, observe that type $T$ agent can only result in outcome $T$ regardless of which action is played. Thus, contract $(x_1,\dots,x_{T})$ pays type $T$ agent $x_T$ for nothing in return.

Altogether, the profit of contract $(x_1,\dots,x_{T})$ in expectation over agent types is at most
\begin{align*}
\sum_{k\in[N],\ell\in[2^N]} \frac{2^{k-1}}{16^N}\cdot \frac{2^{-k}}{2T(1-2^{-k})} - (1-\frac{2^{N}-1}{8^N})\cdot x_T&\le\sum_{k\in[N],\ell\in[2^N]} \frac{2^{k}}{16^N}\cdot \frac{2^{-k}}{2T} - \frac{1}{2}\cdot x_T \\
&=\frac{N}{8^N\cdot 2T}-\frac{1}{2}\cdot x_T,
\end{align*}
and for this to be non-negative, it must hold that $x_T\le\frac{N}{8^N\cdot T}\le\frac{1}{16^N}$.
\end{proof}

\begin{lemma}\label{lemma:bound_lower_part_agents_Omega_log_T_gap}
Consider the principal-agent problem $\PrincipalAgentProblem$ defined in the above $\Omega(\log T)$-gap instance. For any $d\in[N]$, for any contract $\ContractVector=(x_1,\dots,x_{T})\in\mathbb{R}_{\ge 0}^T$ with $x_1\ge1-2^{-d+1}$, we have
\begin{equation}\label{eq:bound_upper_part_agents_Omega_log_T_gap}
    \sum_{k\in[d],\ell\in[2^N]} \Pr[\textrm{agent type}=(k-1)2^N+\ell]\cdot\TypedProfit{(k-1)2^N+\ell}(\CostVector, \ForecastTensor, \RewardVector, \ContractVector)\le\frac{2}{8^N\cdot T}.
\end{equation}
\end{lemma}
\begin{proof}
Note that only outcome 1 has strictly positive reward, and only action $1$ can result in outcome 1. Moreover, since $F^{(t)}_{1,1}\le \frac{1}{T}$ for all $t\in[T]$, outcome 1 happens with probability at most $\frac{1}{T}$ regardless of the contract and the agent type. Thus, contract $(x_1,\dots,x_{T})$ gets profit at most $\frac{1-x_1}{T}$ (which in turn is at most $\frac{2^{-d+1}}{T}$ by our assumption) regardless of the agent type. It follows that
\begin{align*}
    &\sum_{k\in[d],\ell\in[2^N]} \Pr[\textrm{agent type}=(k-1)2^N+\ell]\cdot\TypedProfit{(k-1)2^N+\ell}(\CostVector, \ForecastTensor, \RewardVector, \ContractVector)\\
    \le&\sum_{k\in[d],\ell\in[2^N]} \Pr[\textrm{agent type}=(k-1)2^N+\ell]\cdot\frac{2^{-d+1}}{T}\\
    =&\sum_{k\in[d],\ell\in[2^N]}\frac{2^{k-1}}{16^N}\cdot\frac{2^{-d+1}}{T}
    =\frac{2(1-2^{-d})}{8^N\cdot T}\le\frac{2}{8^N\cdot T}.
\end{align*}
\end{proof}

\begin{lemma}\label{lemma:bound_upper_part_agents_Omega_log_T_gap}
Consider the principal-agent problem $\PrincipalAgentProblem$ defined in the above $\Omega(\log T)$-gap instance. For any $d\in[N]$, given any contract $\ContractVector=(x_1,\dots,x_{T})\in\mathbb{R}_{\ge 0}^T$ with $x_1\in[1-2^{-d+1},1-2^{-d}]$ that has non-negative profit, we have
\[
|\{t\in\{d\cdot2^N+1,\dots,T-1\} \textrm{ s.t. type $t$ agent prefers playing action $1$}\}|\le 2^{N+2}+1.
\]
\end{lemma}
\begin{proof}
Consider any $t\in\{d\cdot2^N+1,\dots,T-1\}$ such that type $t$ agent prefers playing action $1$ given contract $(x_1,\dots,x_{T})$. First, since type $t$ agent prefers action $1$ over the null action, action $1$ must have non-negative utility. Because for type $t$ agent, action 1 only has strictly positive probability for outcomes $1$, $t$ and $T$, it follows that $F^{(t)}_{1,1}\cdot x_1+F^{(t)}_{1,t}\cdot x_t+F^{(T)}_{1,T}\cdot x_T\ge c_1$. After rearrangement, this is equivalent to
\begin{align}\label{eq:IR_Omega_log_T_gap}
   x_t&\ge\frac{c_1-F^{(t)}_{1,1}\cdot x_1-F^{(T)}_{1,T}\cdot x_T}{F^{(t)}_{1,t}}\nonumber\\
   &\ge\frac{c_1-F^{(t)}_{1,1}\cdot x_1-1/16^N}{F^{(t)}_{1,t}}&&\text{(By Lemma~\ref{lemma:x_T_upper_bound_Omega_log_T_gap} and $F^{(T)}_{1,T}\le 1$)}\nonumber\\
   &=\frac{1/(2T)}{F^{(t)}_{1,1}}-x_1-\frac{1/16^N}{F^{(t)}_{1,1}}&&\text{(By definition of $F^{(t)}_{1,t}$ and $c_1$)}\nonumber\\
   &\ge1-2^{-d-1}-x_1-\frac{2T}{16^N}&&\text{($\frac{1}{2T}\le F^{(t)}_{1,1}\le\frac{1}{2T(1-2^{-d-1})}$ for $t\ge d\cdot 2^N+1$)}\nonumber\\
   &\ge2^{-d-1}-\frac{2T}{16^N}&&\text{(By assumption $x_1\le1-2^{-d}$)}\nonumber\\
   &\ge\frac{1}{2^{N+1}}-\frac{2N}{8^N}\ge\frac{1}{2^{N+2}}&&\text{(By $d\le N$ and $T\ge N\cdot2^N$ and $N\ge 3$)}.
\end{align}

Moreover, since type $t$ agent prefers action $1$ over action $2$, action $1$ must have higher utility than action $2$. Namely, $F^{(t)}_{1,1}\cdot x_1+F^{(t)}_{1,t}\cdot x_t+F^{(T)}_{1,T}\cdot x_T-c_1\ge\sum_{\substack{s\in[T-1]\\s\neq t,1}} F^{(t)}_{2,s} \cdot x_{s}+F^{(t)}_{2,T}\cdot x_T$ (recall $F^{(t)}_{2,1},F^{(t)}_{2,t}=0$). After rearrangement, this is equivalent to
\begin{align}\label{eq:IC1_Omega_log_T_gap}
    x_t&\ge \frac{c_1-F^{(t)}_{1,1}\cdot x_1+(F^{(t)}_{2,T}-F^{(t)}_{1,T})\cdot x_T}{F^{(t)}_{1,t}}+\sum_{\substack{s\in[T-1]\\ s\neq t,1 }} \frac{F^{(t)}_{2,s}}{F^{(t)}_{1,t}}\cdot x_{s}\nonumber\\
    &\ge \frac{c_1-F^{(t)}_{1,1}\cdot x_1-1/16^N}{F^{(t)}_{1,t}}+\sum_{\substack{s\in[T-1]\\ s\neq t,1 }} \frac{F^{(t)}_{2,s}}{F^{(t)}_{1,t}}\cdot x_{s} &&\text{(By Lemma~\ref{lemma:x_T_upper_bound_Omega_log_T_gap} and $F^{(t)}_{2,T}-F^{(t)}_{1,T}\ge -1$)}\nonumber\\
    &\ge \frac{c_1-F^{(t)}_{1,1}\cdot x_1-1/16^N}{F^{(t)}_{1,t}}+\sum_{\substack{s\in[T-1]\\ s\neq t,1 }} \frac{1}{2^{N+1}}\cdot x_{s} &&\text{(By definition of $F^{(t)}_{2,s}$)}\nonumber\\
    &\ge\frac{1}{2^{N+2}}+\sum_{\substack{s\in[T-1]\\ s\neq t,1 }} \frac{1}{2^{N+1}}\cdot x_{s},
\end{align}
where the final step follows from the second to the last steps of the derivation of Ineq.~\eqref{eq:IR_Omega_log_T_gap}.

Now let $S_1=\{t\in\{d\cdot2^N+1,\dots,T-1\} \textrm{ s.t. type $t$ agent prefers playing action $1$}\}$, and assume for contradiction $|S_1|\ge 2^{N+2}+2$. By Ineq.~\eqref{eq:IC1_Omega_log_T_gap}, we have that for all $t\in S_1$
\begin{align}\label{eq:IC2_Omega_log_T_gap}
    x_t&\ge\frac{1}{2^{N+2}}+\sum_{\substack{s\in[T-1]\\ s\neq t,1 }} \frac{1}{2^{N+1}}\cdot x_{s}\nonumber\\
    &\ge\sum_{\substack{s\in S_1\\ s\neq t,1 }} \frac{1}{2^{N+1}}\cdot x_{s}\nonumber\\
    &\ge(|S_1|-2)\cdot \frac{1}{2^{2N+3}} &&\text{(By Ineq.~\eqref{eq:IR_Omega_log_T_gap})}\nonumber\\
    &\ge\frac{1}{2^{N+1}} &&\text{(By assumption $|S_1|\ge 2^{N+2}+2$)}.
\end{align}

Notice that Ineq.~\eqref{eq:IC2_Omega_log_T_gap} gives a new lower bound of $x_t$ for all $t\in S_1$, which is twice as large as the lower bound given by Ineq.~\eqref{eq:IR_Omega_log_T_gap}. Using this larger lower bound, we can repeat the above derivation of Ineq.~\eqref{eq:IC2_Omega_log_T_gap} (i.e., we use the larger lower bound instead in the third inequality of the derivation) and double the lower bound again. By repeating such derivation arbitrarily many times, we have that $x_t$ for all $t\in S_1$ is arbitrarily large, and thus, contract $\ContractVector$ makes arbitrarily large payment to each type of agent in $S_1$. Therefore, contract $\ContractVector$ must have strictly negative profit (because the reward generated by any type of agent is always bounded), which contradicts our assumption in the lemma statement.
\end{proof}

\begin{proof}[Proof of Theorem~\ref{thm:Omega_log_T_gap}]
\textbf{Soundness.} We show that any contract $\ContractVector=(x_1,\dots,x_{T})\in\mathbb{R}_{\ge 0}^T$ with non-negative profit at most has expected profit $\frac{6}{N\cdot16^N}$ (the expectation is over the distribution of agent type). Let $d\in[N]$ be such that $x_1\in[1-2^{-d+1},1-2^{-d}]$. Now we upper bound the expected profit of contract $\ContractVector$ as follows:
\begin{align}\label{eq:soundness1_Omega_log_T_gap}
    \textrm{expected profit}&=\sum_{k\in[N],\ell\in[2^N]} \Pr[\textrm{agent type}=(k-1)2^N+\ell]\cdot\TypedProfit{(k-1)2^N+\ell}(\CostVector, \ForecastTensor, \RewardVector, \ContractVector)\nonumber\\
    &=\sum_{k\in[d],\ell\in[2^N]} \Pr[\textrm{agent type}=(k-1)2^N+\ell]\cdot\TypedProfit{(k-1)2^N+\ell}(\CostVector, \ForecastTensor, \RewardVector, \ContractVector)\nonumber\\
    &\qquad+\sum_{t\in\{d\cdot2^N+1,\dots N\cdot2^N\}} \Pr[\textrm{agent type}=t]\cdot\TypedProfit{t}(\CostVector, \ForecastTensor, \RewardVector, \ContractVector)\nonumber\\
    &\le\frac{2}{8^N\cdot T}+\sum_{t\in\{d\cdot2^N+1,\dots N\cdot2^N\}} \Pr[\textrm{agent type}=t]\cdot\TypedProfit{t}(\CostVector, \ForecastTensor, \RewardVector, \ContractVector)\,\,\,\,\text{(By Lemma~\ref{lemma:bound_lower_part_agents_Omega_log_T_gap})}.
\end{align}
Moreover, because only action $1$ can generate strictly positive reward, we have that
\begin{align}\label{eq:soundness2_Omega_log_T_gap}
    &\sum_{t\in\{d\cdot2^N+1,\dots N\cdot2^N\}} \Pr[\textrm{agent type}=t]\cdot\TypedProfit{t}(\CostVector, \ForecastTensor, \RewardVector, \ContractVector)\nonumber\\
    &=\sum_{\substack{t\in\{d\cdot2^N+1,\dots N\cdot2^N\}\\\textrm{ s.t. type $t$ agent prefers action $1$}}} \Pr[\textrm{agent type}=t]\cdot\TypedProfit{t}(\CostVector, \ForecastTensor, \RewardVector, \ContractVector)\nonumber\\
    &\le\sum_{\substack{t\in\{d\cdot2^N+1,\dots N\cdot2^N\}\\\textrm{ s.t. type $t$ agent prefers action $1$}}} \Pr[\textrm{agent type}=t]\cdot \textrm{expected welfare of action $1$ for type $t$ agent}\nonumber\\
    &=\sum_{\substack{t\in\{d\cdot2^N+1,\dots N\cdot2^N\}\\\textrm{ s.t. type $t$ agent prefers action $1$}}} \Pr[\textrm{agent type}=t]\cdot (F^{(t)}_{1,1}\cdot 1-c_1)\nonumber\\
    &\le\sum_{\substack{t\in\{d\cdot2^N+1,\dots N\cdot2^N\}\\\textrm{ s.t. type $t$ agent prefers action $1$}}} \frac{1}{2T\cdot 16^N} \qquad\text{(By Ineq.~\eqref{eq:equal_revenue_Omega_log_T_gap})}\nonumber\\
    &\le \frac{2^{N+2}+1}{2T\cdot 16^N}\le \frac{4}{T\cdot 8^N} \qquad\qquad\qquad\qquad\quad\,\text{(By Lemma~\ref{lemma:bound_upper_part_agents_Omega_log_T_gap})}.
\end{align}
Combining Ineq.~\eqref{eq:soundness1_Omega_log_T_gap} and Ineq.~\eqref{eq:soundness2_Omega_log_T_gap}, we get that the expected profit of contract $\ContractVector$ is at most $\frac{6}{T\cdot 8^N}\le\frac{6}{N\cdot16^N}$.

\textbf{Completeness.} We show that there is a menu of $T-1$ contracts $\{\ContractVector_1,\dots,\ContractVector_{T-1}\}$ that achieve expected profit $\frac{1}{16^{N+1}}$. Specifically, for each $k\in[N]$ and $\ell\in[2^N]$, $\ContractVector_{(k-1)2^N+\ell}:=(x_1^{(k,\ell)},x_2^{(k,\ell)},\dots,x_{T}^{(k,\ell)})$ is defined as follows: let $x_1^{(k,\ell)}=\frac{1}{2}$, and let $x_{(k-1)2^N+\ell}^{(k,\ell)}=\frac{1}{2}-\frac{1}{2^{k+1}}$ if $(k-1)2^N+\ell\neq1$, and let the other entries be zero.

Because all the contracts in the menu make zero payment for outcome $T$, and type $T$ agent can only result in outcome $T$ which has zero reward, we can ignore type $T$ agent.

Now we show that given the above menu, for all $k\in[N]$ and $\ell\in[2^N]$, type $((k-1)2^N+\ell)$ agent wants to pick contract $\ContractVector_{(k-1)2^N+\ell}$ and play action $1$. First, it is easy to check that for type $((k-1)2^N+\ell)$ agent, given contract $\ContractVector_{(k-1)2^N+\ell}$, action $2$ generates zero utility, and action $1$ generates utility
\begin{equation}\label{eq:utility_completeness_Omega_log_T_gap}
    \frac{1}{2}\cdot F^{((k-1)2^N+\ell)}_{1,1}+\left(\frac{1}{2}-\frac{1}{2^{k+1}}\right)\cdot F^{((k-1)2^N+\ell)}_{1,(k-1)2^N+\ell}-c_1=\left(1-\frac{1}{2^{k+1}}\right)\cdot F^{((k-1)2^N+\ell)}_{1,1}-c_1=\frac{2^{-k-2}}{T(1-2^{-k})},
\end{equation}
which is strictly positive. Thus, it remains to show that by picking any other contract $\ContractVector_{(k'-1)2^N+\ell'}$, type $((k-1)2^N+\ell)$ agent cannot get more than the utility in Eq.~\eqref{eq:utility_completeness_Omega_log_T_gap}. To this end, we notice that given contract $\ContractVector_{(k'-1)2^N+\ell'}$, type $((k-1)2^N+\ell)$ agent would get the following utilities by playing action $1$ and action $2$ respectively:
\begin{align*}
    \textrm{utility of action $1$}&=\frac{1}{2}\cdot F^{((k-1)2^N+\ell)}_{1,1}-c_1\le 0,\\
    \textrm{utility of action $2$}&=(\frac{1}{2}-\frac{1}{2^{k'+1}})\cdot F^{((k-1)2^N+\ell)}_{2,(k'-1)2^N+\ell}\le \frac{1}{2}\cdot F^{((k-1)2^N+\ell)}_{2,(k'-1)2^N+\ell}=\frac{2^{-k-3}}{T(1-2^{-k})}.
\end{align*}

Finally, we derive the expected profit of the above menu (given that we have shown type $((k-1)2^N+\ell)$ agent picks contract $\ContractVector_{(k-1)2^N+\ell}$ and plays action $1$) as follows:
\begin{align*}
    \textrm{expected profit}&= \sum_{k\in[N],\ell\in[2^N]}\Pr[\textrm{agent type}=(k-1)2^N+\ell]\\
    &\qquad\qquad\qquad\times\left(F^{((k-1)2^N+\ell)}_{1,1}\cdot \left(1-x_1^{(k,\ell)}\right)-F^{((k-1)2^N+\ell)}_{1,(k-1)2^N+\ell}\cdot x_{(k-1)2^N+\ell}^{(k,\ell)}\right)\\
    &=\sum_{k\in[N],\ell\in[2^N]}\Pr[\textrm{agent type}=(k-1)2^N+\ell]\cdot F^{((k-1)2^N+\ell)}_{1,1}\cdot \left(1-x_1^{(k,\ell)}- x_{(k-1)2^N+\ell}^{(k,\ell)}\right)\\
    &=\sum_{k\in[N],\ell\in[2^N]}\frac{2^{k-1}}{2T(1-2^{-k})\cdot16^N}\cdot\frac{1}{2^{k+1}}=\sum_{k\in[N],\ell\in[2^N]}\frac{1}{8T\cdot16^N(1-2^{-k})}\\
    &=\frac{T-1}{8T\cdot16^N(1-2^{-k})}\ge\frac{1}{ 16^{N+1}}.
\end{align*}
\end{proof}

\subsection{Bounding the gap between linear contracts and menus of randomized linear contracts}

\begin{lemma} \label{thm:Omega_T_lower}
The optimal linear contract achieves a $O(T)$ approximation with respect to the optimal menu of randomized linear contracts. In particular, 
  \begin{align*}
    \OLinear \PrincipalAgentProblem \ge \Omega(1/T) \cdot \ORndMenuLinear \PrincipalAgentProblem. 
  \end{align*}
\end{lemma}
\begin{proof}
For each type, we can compute the best linear contract. This contract generates at least as much revenue as any component of the randomized linear contract. Hence, we can extract at least an $\Omega(1/T)$ fraction of the optimal revenue. 
\end{proof}

\subsection{Proof of Lemma~\ref{lem:brkpoint}}
\label{subsec:apx-lem-brkpoint}

\begin{proof}[Proof of Lemma~\ref{lem:brkpoint}]
 Given any optimal menu $\tilde{\Gamma} = \left(\tilde{\gamma}^{(1)},\dots,\tilde{\gamma}^{(T)} \right)$, we will construct a new menu $\Gamma=\left(\gamma^{(1)},\dots,\gamma^{(T)} \right)$ which satisfies the above constraints and whose revenue is at least as large as $\tilde{\Gamma}$. In particular, we will demonstrate how to redistribute probability mass not supported on one of these points to these points while leaving the performance of the menu unchanged. 

 For any point $\alpha \leq \alpha_p$ that is not a break point, we can write $\alpha = r\cdot \alpha_j + (1-r)\cdot \alpha_{j+1}$ as a convex combination of two adjacent breakpoints. Now we simply increment $\gamma^{(t)}(\alpha_j) $ by $r\cdot \tilde{\gamma}^{(t)}(\alpha)$ and $\gamma^{(t+1)}(\alpha_{j+1})$ by$(1-r) \cdot \tilde{\gamma}^{(t)}(\alpha_{j+1})$. Since we are simply moving the mass at a given point to its two neighbors, and all utility functions $U_t$ are linear in between any pair of consecutive breakpoints, the utility of each agent for each randomized linear contract in the menu remains the same (and hence each agent still selects the same contract from the menu and plays the same action).
 
 The profit obtained by the principal by offering linear contract $\alpha$ to the agent of type $t$ is given by $(1-\alpha)U'_{t}(\alpha)$ (when taking the derivative, we arbitrarily break ties in the favor of the principal). Since $U_i'(\alpha) = U_i'(\alpha_j) \leq U_i'(\alpha_{j+1})$, we note that the revenue achieved by the new menu is at least what is achieved by $\tilde{\Gamma}$.

It remains to deal with the probability mass above $\alpha_{p}$. Define the point 
$$\alpha^{i}_{p+1}:= \frac{ \int_{b=\alpha_p}^{\infty} b \cdot \tilde{\gamma}^{(i)}(b) db}{\int_{b=\alpha_p}^{\infty}\tilde{\gamma}^{(i)}(b) db}$$ 
whenever the $\int_{b=\alpha_p}^{\infty}\gamma^{(i)}(b) db >0$.  By definition,  $\alpha^i_{p+1} \geq \alpha_p$. We will set $\gamma^{(i)}(\alpha^i_{p+1}) = \int_{b=\alpha_p}^{\infty}\gamma^{(i)}(b) db $.
Since we are simply contracting the mass in this last interval $(\alpha_p,\infty)$, the resulting distribution is still a probability distribution. 
Furthermore, 
$$\int_{b=\alpha_{p}}^{\infty} f(b) \tilde{\gamma}^{(i)}(b) db = f(\alpha_{p+1}) \gamma^{(i)}(\alpha_{p+1})$$
for any linear function $f$. Again, since (for this range of $\alpha$) all utilities for each agent (and the principal's profit) are linear functions of $\tilde{\gamma}$, they are preserved under the new contract $\gamma$.

\end{proof}
\subsection{Proof of Lemma~\ref{lem:unbounded_rlc_example}}

\begin{proof}[Proof of~\Cref{lem:unbounded_rlc_example}]
Recall that the breakpoints in the instance are multiples of $1/3$, i.e.~$\alpha_i = i/3$ for $i \leq 6.$
First, we exhibit a menu whose revenue is large $0.31625$. Consider the menu $ \Gamma = (\gamma^{(1)}, \gamma^{(2)})$
\begin{align*}
  \gamma^{(1)}(\alpha) = \begin{cases}
    \frac{3}{5}    &\text{ if } \alpha = 0  \\
    \frac{79}{200} &\text{ if } \alpha = \frac23 \\
    \frac{1}{200}  &\text{ if } \alpha = \frac53 \\
    0              &\text{ otherwise}
  \end{cases}
\qquad 
  \gamma^{(2)}(\alpha) = \begin{cases}
    1 &\text{ if } \alpha = \frac13 \\
    0 &\text{ otherwise} 
  \end{cases}
\end{align*}
The first menu item extracts a revenue of $\frac{359}{2400}$ and the second extracts a revenue of $\frac16$ to get a total revenue of $\frac{253}{800}=0.31625$.

To show that no bounded randomized linear contract cannot get a fractional value of greater than $0.31625$, we exhibit a dual solution to the dual of linear program presented in Section~\ref{sec:menu_rlc}.
\begin{align*}
& \min & s_1 + s_2 \\
& s.t. &
  s_1 &\ge \left( U_1(\alpha_j) \lambda_{1,2} - U_{2}(\alpha_j) \lambda_{2,1} \right) + U_{1}'(\alpha_j) (1-\alpha_j)  \qquad \forall 0 \le j \le 3  \\
& &  s_2 &\ge \left( U_2(\alpha_j) \lambda_{2,1} - U_{1}(\alpha_j) \lambda_{1,2} \right) + U_{2}'(\alpha_j) (1-\alpha_j)  \qquad \forall 0 \le j \le 3  \\
& & \lambda, s &\geq 0
\end{align*}

We derived this dual by restricting the breakpoints to just $\alpha_0 = 0, \alpha_1 = 1/3, \alpha_2 = 2/3, \alpha_3 = 1$. Hence any feasible dual solution is an upper bound on menus of \emph{bounded} randomized linear contracts. In particular, assigning values $\lambda_{1,2} = 1.857154080498256 $ , $\lambda_{2,1} = 0.8095306797975786$, $s_1 = 0.25$ and $s_2 = 0.056345$ achieves an objective of approximately $0.30635<0.31$. 
\end{proof}

\subsection{Proof of Corollary~\ref{cor:det_vs_rnd}}\label{apx:det_vs_rnd}

\begin{proof}[Proof of Corollary~\ref{cor:det_vs_rnd}]
Note that the example in Theorem \ref{thm:rlc_gap} only has one non-null outcome. We first claim that for such instances, every general contract is equivalent to an \textit{affine contract}: a contract that transfers $\alpha$ times the principal's reward plus an additive $\beta$ (in particular, a transfer of $x_0$ on the null outcome and $x_1$ on the non-null outcome is equivalent to the affine contract with $\beta = x_0$ and $\alpha = (x_1 - x_0)/r_1$). 

Because of this, the best menu of randomized (general) contracts performs equally as well as the best menu of randomized affine contracts. Since all linear contracts are affine, this menu performs at least as well as the best menu of randomized linear contracts. We thus have that

$$\ORndMenu \PrincipalAgentProblem \geq \ORndMenuLinear\PrincipalAgentProblem.$$

On the other hand, Lemma 7 in \citet{guruganesh2021contracts} states that for typed principal-agent problems with two outcomes, $\ODetMenu\PrincipalAgentProblem = \OLinear\PrincipalAgentProblem$. Combining these two facts with Theorem \ref{thm:rlc_gap} we arrive at the theorem statement.
\end{proof}

\subsection{Proof of Theorem~\ref{thm:Omega_1_lower}}\label{sec:Omega_1_lower}

\begin{proof}[Proof of Theorem~\ref{thm:Omega_1_lower}]
Consider any $\PrincipalAgentProblem{}$ with $n = 2$ actions (a null action $0$ and a non-null action $1$), $T$ agent types, and $m$ outcomes. For each agent type $t$, we let $\Utility_t(\alpha):\mathbb{R}_{\geq 0} \rightarrow \mathbb{R}_{\geq 0}$ denote the utility function for type-$t$ agent given linear contract $\alpha$ (i.e., pay $\alpha \cdot r_j$ for each outcome $j$) as input variable. Notice that $\Utility_t(\alpha)=\max\{0,\,\sum_{j\in[m]} \alpha\cdot F_{1,j}^{(t)} r_j-c_1\}$ (i.e., the best of the utilities of actions $0$ and $1$). We let $R^{(t)}:=\sum_{j\in[m]} F_{1,j}^{(t)} r_j$. Then, the utility function can be simplified as $\Utility_t(\alpha)=\max\{0,\,R^{(t)}\cdot \alpha-c_1\}$. For each agent type $t$, we let $\alpha_t$ denote the minimum linear contract for which type-$t$ agent would like to play action $1$, i.e., $\alpha_t=\frac{c_1}{R^{(t)}}$.

\paragraph{Bucketizing the agent types}
Let $w=10$ (any constant larger than $1$ would work) and $\alpha_{\min}:=\min_{t\in[T]}\alpha_t$ and $\alpha_{\max}:=\max_{t\in[T]}\alpha_t$. Now we bucketize the agent types into $b$ buckets according to their $\alpha_t$, and $b$ is chosen such that $1-\alpha_{\min}\in[w^{b-1}(1-\alpha_{\max}), w^{b}(1-\alpha_{\max}))$. The $i$-th bucket $B_i$ contains every agent type $t$ such that $1-\alpha_t\in[w^{i-1}(1-\alpha_{\max}), w^{i}(1-\alpha_{\max}))$. Now consider two principal-agent problem instances by modifying the instance $\PrincipalAgentProblem{}$:
\begin{itemize}
    \item Odd instance $(\CostVector_1, \ForecastTensor_1, \RewardVector_1)$: Remove every agent type that belongs to a bucket with even index. Add a dummy agent type (which always results into the dummy outcome with zero reward) such that its probability is the sum of the probabilities of the removed agent types.
    \item Even instance $(\CostVector_2, \ForecastTensor_2, \RewardVector_2)$: Same as above except that we now remove the agent types in the buckets with odd indices.
\end{itemize}

Obviously $\OLinear(\CostVector, \ForecastTensor, \RewardVector)\ge\OLinear(\CostVector_1, \ForecastTensor_1, \RewardVector_1),\OLinear(\CostVector_2, \ForecastTensor_2, \RewardVector_2)$, and
\begin{align*}
    \ORndMenuLinear(\CostVector, \ForecastTensor, \RewardVector)\le &\ORndMenuLinear(\CostVector_1, \ForecastTensor_1, \RewardVector_1)\\
    &+ \ORndMenuLinear(\CostVector_2, \ForecastTensor_2, \RewardVector_2).
\end{align*}

Assume that $\ORndMenuLinear(\CostVector_1, \ForecastTensor_1, \RewardVector_1)\ge\ORndMenuLinear(\CostVector_2, \ForecastTensor_2, \RewardVector_2)$ without loss of generality. Then, we have that
\[
\frac{\OLinear(\CostVector, \ForecastTensor, \RewardVector)}{ \ORndMenuLinear(\CostVector, \ForecastTensor, \RewardVector)}\ge\frac{\OLinear(\CostVector_1, \ForecastTensor_1, \RewardVector_1)}{2\cdot\ORndMenuLinear(\CostVector_1, \ForecastTensor_1, \RewardVector_1)}.
\]
Henceforth, it suffices to show $\OLinear(\CostVector_1, \ForecastTensor_1, \RewardVector_1)\ge\Omega(1)\cdot\ORndMenuLinear(\CostVector_1, \ForecastTensor_1, \RewardVector_1)$. In the principal-agent problem $(\CostVector_1, \ForecastTensor_1, \RewardVector_1)$, if two agent types $t,t'$ belong to two different buckets (say $t\in B_i,t'\in B_{i'}$ for $i<i'$), then $i+1<i'$ because $i,i'$ must be odd. It follows that $w(1-\alpha_{t})\le 1-\alpha_{t'}$ by definition of the buckets.

\paragraph{Choosing one agent type from each bucket} Notice that for each bucket $B_i$, linear contract $\beta_i:=1-w^{i-1}(1-\alpha_{\max})$ can extract $\frac{1}{w}$-fraction of the welfare of any agent type $t\in B_i$. Indeed, for any $t\in B_i$, we have $\beta_i\ge \alpha_t$, and hence type-$t$ agent would play action $1$ and generates a profit $(1-\beta_i)R^{(t)}$ for the principal, which is $\frac{1-\beta_i}{1-\alpha_t}\ge\frac{1-\beta_i}{w^{i}(1-\alpha_{\max})}=\frac{1}{w}$ fraction of type-$t$ agent's welfare.

Suppose the optimal linear contract only achieves expected profit $p_{\ell}$ for $(\CostVector_1, \ForecastTensor_1, \RewardVector_1)$. Then the expected welfare of bucket $B_i$ (i.e., $\sum_{t\in B_i} \Pr[\textrm{agent type}=t]\cdot\textrm{type-$t$ agent's welfare}$) for any odd $i$ is at most $w\cdot p_{\ell}$, because we have shown there is a linear contract that extracts $\frac{1}{w}$-fraction of the welfare of any agent type $t\in B_i$. Moreover, suppose the optimal menu of randomized linear contracts $\Gamma$ extracts $\eta_i$ fraction of the expected welfare of bucket $B_i$. Then the expected profit of $\Gamma$ is at most $\sum_{\textrm{odd } i\in[b]}\eta_i w\cdot p_{\ell}$. Since $w$ is a constant, it suffices to prove $\sum_{\textrm{odd } i\in[b]}\eta_i=O(1)$.

Moreover, since $\Gamma$ extracts $\eta_i$ fraction of the expected welfare of bucket $B_i$, there should be at least one agent type in $B_i$ from which $\Gamma$ extracts at least $\eta_i$ fraction of the welfare.
We choose one such agent type for each bucket $B_i$ with odd $i\in[b]$.

Summarizing what we have done so far, we have chosen $k=\left\lceil\frac{b}{2}\right\rceil$ agent types, which we denote by $t_1,t_2,\dots,t_k$ (ordered s.t. $\alpha_{t_i}$ is non-decreasing in $i$), such that
\begin{itemize}
    \item $w(1-\alpha_{t_{i+1}})\le 1-\alpha_{t_i}$ for each $i\in[k-1]$ (recall this follows from picking the odd instance),
    \item and menu $\Gamma$ extracts certain $\rho_i$ fraction of the welfare of the type-$t_i$ agent for each $t_i$, such that $\sum_{i\in[k]}\rho_i\ge\sum_{\textrm{odd } i\in[b]}\eta_i$ (this is by our choice of the agent types),
\end{itemize}
and our final step is to show $\sum_{i\in[k]}\rho_i=O(1)$, which implies $\sum_{\textrm{odd } i\in[b]}\eta_i=O(1)$.

\paragraph{Proving that $\sum_{i\in[k]}\rho_i=O(1)$}
Let $D_i(\alpha)$ denote the CDF of type-$t_i$ agent's favorite randomized linear contract for each $i\in[k]$ (we will also use $D_i$ to refer to type-$t_i$ agent's favorite randomized linear contract). Then, we have that
\begin{align*}
    &\underbrace{\frac{\textrm{profit of $D_i$ from type-$t_i$ agent}}{\textrm{type-$t_i$ agent's welfare}}}_{=\rho_i}=\int_{\alpha_{t_i}}^{\infty} \frac{1-\alpha}{1-\alpha_{t_i}} dD_i(\alpha),\\
    &\underbrace{\frac{\textrm{type-$t_i$ agent's utility from $D_i$}}{\textrm{type-$t_i$ agent's welfare}}}_{:=u_i}=\int_{\alpha_{t_i}}^{\infty} \frac{\alpha-\alpha_{t_i}}{1-\alpha_{t_i}} dD_i(\alpha),
\end{align*}
where the integrals start from $\alpha_{t_i}$ because $\alpha_{t_i}$ is the minimum $\alpha$ for which type-$t_i$ agent would like to play action $1$. Note that the LHS of the first equation is just $\rho_i$ by definition, and we let $u_i$ denote the LHS of the second equation (observe that $\rho_i+u_i=1$). Without loss of generality, we assume that $u_i\le 1$ for all $i\in[k]$, because otherwise we can first remove all the agent types $t_i$ with $u_i>1$ and then prove $\sum_{i\in[k]\textrm{ s.t. }u_i\le1}\rho_i=O(1)$ for the remaining agent types (which obviously implies $\sum_{i\in[k]}\rho_i=O(1)$).

Now we establish a lower bound for $u_{i-1}-u_i$. Because type-$t_{i-1}$ agent prefers $D_{i-1}$ over $D_i$, we have that
\[
\underbrace{\frac{\textrm{type-$t_{i-1}$ agent's utility from $D_{i-1}$}}{\textrm{type-$t_{i-1}$ agent's welfare}}}_{=u_{i-1}}\ge \underbrace{\frac{\textrm{type-$t_{i-1}$ agent's utility from $D_{i}$}}{\textrm{type-$t_{i-1}$ agent's welfare}}}_{:=u_{i-1}'},
\]
where we denote the RHS by $u_{i-1}'$. Thus we can prove a lower bound for $u_{i-1}'-u_i$ instead. To this end, we derive that
\begin{align*}
    u_{i-1}'-u_i&=\int_{\alpha_{t_{i-1}}}^{\infty} \frac{\alpha-\alpha_{t_{i-1}}}{1-\alpha_{t_{i-1}}} dD_i(\alpha)-\int_{\alpha_{t_i}}^{\infty} \frac{\alpha-\alpha_{t_i}}{1-\alpha_{t_i}} dD_i(\alpha)\\
    &\ge \int_{\alpha_{t_i}}^{\infty} \frac{\alpha-\alpha_{t_{i-1}}}{1-\alpha_{t_{i-1}}} dD_i(\alpha)-\int_{\alpha_{t_i}}^{\infty} \frac{\alpha-\alpha_{t_i}}{1-\alpha_{t_i}} dD_i(\alpha) &&\text{(By $\alpha_{t_i}\ge \alpha_{t_{i-1}}$)}\\
    &=\int_{\alpha_{t_i}}^{\infty}\frac{(1-\alpha_{t_i})(\alpha-\alpha_{t_{i-1}})-(1-\alpha_{t_{i-1}})(\alpha-\alpha_{t_i})}{(1-\alpha_{t_i})(1-\alpha_{t_{i-1}})}dD_i(\alpha)\\
    &=\int_{\alpha_{t_i}}^{\infty}\frac{-\alpha_{t_{i-1}}-\alpha_{t_i}\alpha+\alpha_{t_i}+\alpha_{t_{i-1}}\alpha}{(1-\alpha_{t_i})(1-\alpha_{t_{i-1}})}dD_i(\alpha)\\
    &=\int_{\alpha_{t_i}}^{\infty} \frac{1-\alpha}{1-\alpha_{t_i}}\cdot\frac{\alpha_{t_i}-\alpha_{t_{i-1}}}{1-\alpha_{t_{i-1}}} dD_i(\alpha)\\
    &\ge \int_{\alpha_{t_i}}^{\infty} \frac{1-\alpha}{1-\alpha_{t_i}}\cdot\left(1-\frac{1}{w}\right) dD_i(\alpha)&&\text{(By $w(1-\alpha_{t_i})\le(1-\alpha_{t_{i-1}})$)}\\
    &=\left(1-\frac{1}{w}\right)\rho_i &&\text{(By definition of $\rho_i$)}.
\end{align*}
Therefore, we have shown that $u_{i-1}-u_i\ge \left(1-\frac{1}{w}\right)\rho_i$. Taking a telescoping sum, we have that $u_{1}-u_{k}=\sum_{i\in\{2,\dots,k\}}u_{i-1}-u_{i}\ge \left(1-\frac{1}{w}\right)\sum_{i\in\{2,\dots,k\}}\rho_i$. It follows that $\sum_{i\in[k]}\rho_i=O(1)$ because $\rho_1,u_1\le 1$ and $u_k\ge 0$, and $w$ is a constant.
\end{proof}
\end{document}